\documentclass[letterpaper, 10 pt, journal, twocolumn]{IEEEtran}
\usepackage{graphicx}      % include this line if your document contains figures
\usepackage{tabulary}
\usepackage{amsfonts}
\usepackage{amsmath}
\usepackage{amssymb}

\usepackage{pgf,tikz}
\usepackage{tabularx}
\usepackage{float}
\usepackage{cite}
\usepackage{comment}
\usepackage{stfloats}
\usepackage{mathtools}
\usepackage{mathabx}
\usepackage{url}
\usepackage[ruled,vlined,linesnumbered,resetcount]{algorithm2e}
\usepackage{setspace}

\usepackage{enumitem} % To reference numbered items
\usepackage{bm} 

\newtheorem{theorem}{Theorem}[section]
\newtheorem{proposition}[theorem]{Proposition}
\newtheorem{lemma}[theorem]{Lemma}

\newtheorem{remark}[theorem]{Remark}

\newtheorem{problem}{Problem}

\newtheorem{assumption}{Assumption}

\newcommand{\real}{\mathbb{R}}
\newcommand{\realnonneg}{\mathbb{R}_{\ge 0}}
\newcommand{\col}{{\rm col\;}}
\newcommand{\diag}{{\rm diag\;}}

\newcommand{\image}{\operatorname{image}}
\newcommand*{\horzbar}{\rule[.5ex]{2.5ex}{0.5pt}}

\makeatletter
\renewcommand*\env@matrix[1][\arraystretch]{%
	\edef\arraystretch{#1}%
	\hskip -\arraycolsep
	\let\@ifnextchar\new@ifnextchar
	\array{*\c@MaxMatrixCols c}}
\makeatother

\newcommand\oprocendsymbol{\hbox{$\square$}}
\newcommand\oprocend{\relax\ifmmode\else\unskip\hfill\fi\oprocendsymbol}

\newcommand{\until}[1]{\{1,\dots,#1\}}
\newcommand{\longthmtitle}[1]{\mbox{} \emph{(#1):}}
\newcommand{\marginJC}[1]{\marginpar{\color{red}\tiny\ttfamily#1}}

\allowdisplaybreaks

\begin{document}

\title{Data-Driven Stabilization of Unknown Linear-Threshold Network
  Dynamics
  % Designing Data-Driven Controller for Linear-Threshold Network
  % Dynamics
}

% Conference title: Data-driven control of linear-threshold network
% dynamics

\author{Xuan Wang, Duy Duong-Tran, and Jorge Cort{\'e}s \thanks{A preliminary
    version of this work appeared as \cite{XW-JC:22-acc} at the 2022
    American Control Conference.  X. Wang is with the Department of
    Electrical and Computer Engineering, George Mason University,
    Fairfax, xwang64@gmu.edu. 
    D. Duong-Tran is with the Perelman School of Medicine, 
    University of Pennsylvania, 
    Philadelphia, DuyAnh.Duong-Tran@Pennmedicine.upenn.edu. 
    J. Cort{\'e}s is with the Department of Mechanical and Aerospace
    Engineering, University of California, San Diego,
    cortes@ucsd.edu.\\
    This work was partially supported by NSF ECCS Award 2332210, NSF CMMI Award
2308640, and MURI ARO Award W911NF-24-1-0228.}%
}

%%%%%%%%%%%%%%%%%%%%%%%%%%%%%%% Main Body%%%%%%%%%%%%%%%%%%

\maketitle

\begin{abstract}
  This paper studies the data-driven control of unknown
  linear-threshold network dynamics to stabilize the state to a
  reference value.  We consider two types of controllers: (i) a state
  feedback controller with feed-forward reference input and (ii) an
  augmented feedback controller with error integration.  The first
  controller features a simpler structure and is easier to design,
  while the second offers improved performance in the presence of
  system parameter changes and disturbances.
  % To design these two types of controllers purely based on data
  % samples, we first introduce a map that reconstructs the
  % state-input pair as a transformation of data matrices. Based on
  % this, we establish a data-driven representation of the
  % linear-threshold network in an open-loop form. Using this
  % representation, we incorporate the structures of the two
  % controller types to reformulate the open-loop model and derive
  % their respective closed-loop representations.
  Our design strategy employs state-input datasets to construct
  data-based representations of the closed-loop dynamics. Since these
  representations involve linear threshold functions, we rewrite them
  as switched linear systems, and formulate the design problem as that
  of finding a common controller for all the resulting modes. This
  gives rise to a set of linear matrix inequalities (LMIs) whose
  solutions corresponds to the controller gain matrices.  We analyze
  the computational complexity of solving the LMIs and propose a
  simplified, sufficient set of conditions that scales linearly with
  the system state.  Simulations on two case studies involving
  regulation of firing rate dynamics in rodent brains and of arousal
  level dynamics in humans demonstrate the effectiveness of the
  controller designs.
  % with the error integration-based controller
%   showing improved performance in handling disturbances compared to
%   the feed-forward reference input-based controller.
%
\end{abstract}


\thispagestyle{empty}

\pagestyle{empty}

%===============================================================================

\section{Introduction}

Linear-threshold dynamics are a class of nonlinear dynamics with wide
scientific and practical applications.  These dynamics model
interactions within network systems, where the states of the nodes
evolve based on inter-node connectivity and a linear-threshold
activation function.  Compared to linear systems, linear-threshold
networks provide bounded system states and greater dynamical
versatility, including mono- and multi-stability, limit cycles, and
chaotic behavior. These properties have made them useful in diverse
fields, including computational
neuroscience~\cite{PD-LFA:01,CC-JG-KM:19,EN-JC:21-tacI}, social
networks~\cite{WC-YY-LZ:10,YDZ-VS-NEL:17}, and deep
learning~\cite{HZ-ZW-DL:14,AGH-MZ-BC-DK-WW-TW-MA-HA:17}.  Traditional
approaches to regulate the dynamical behavior of linear-threshold
networks are heavily model-based and rely on precise model parameters
to design control schemes.  Recent advances in data collection,
processing, and computation have motivated a spur of activity on
data-driven control methods for systems with unknown dynamics directly
from its input-output data.
% It is particularly appealing for modern applications, which often
% involve large-scale systems characterized by complex dynamics, where
% obtaining accurate model parameters is challenging.
This paper develops data-driven control methods for stabilization
specifically tailored for linear-threshold network dynamics.

\subsubsection*{Literature review}
The existing literature for data-driven control can be classified into
two major categories: linear and nonlinear systems. For linear
systems, Willems’ Lemma~\cite{JCW-PR-IM-BLMDM:05} has been a
foundational tool in constructing a data-based representation of
input-output trajectories as a linear combination of measured data
samples.  Building on this result, various types of controllers can be
reformulated as part of a closed-loop representation, thereby
facilitating the design of controller parameters. Along this
idea,~\cite{CDP-PT:19} studies the stabilization controller and
linear–quadratic regulator (LQR) of a system by associating the gains
of linear feedback controllers with the solutions of a linear matrix
inequality (LMI) and a semidefinite program (SDP), respectively. The
idea has also been utilized in a receding horizon fashion for model
predictive control (MPC)~\cite{JC-JL-FD:19,JB-JK-MAM-FA:20} and
distributed MPC~\cite{AA-JC:21-csl} when data are only locally
available to nodes of a network. Considering data richness and noise,
\cite{WHJ:21,JB-JK-MAM-FA:20} study the online implementation of
sample efficient data-driven control with noisy system measurements.
Aligned with this body of work, the informativity approach to
data-driven control~\cite{HJVW-JE-HLT-MKC:23} considers measurements
that do not contain enough information to obtain a unique system. By
making assumptions on the model class and noise model, this approach
explicitly determines the set of all systems consistent with the
measurements, thereby enabling the certification of desirable
properties for the measured system. These have included
stability~\cite{HJVW-JE-HLT-MKC:20,USP-MI:09}, controllability and
observability~\cite{HJVW-MKC-HLT:22,ZW-DL:11}, and
dissipativity~\cite{HJVW-MKC-PR-HLT:22,AR-JMM-FA:17}.

%
% \marginJC{Worth alluding to Data informativity papers here too -- we
%   don't want to annoy anybody! E.g.,
%   \url{https://arxiv.org/abs/2302.10488} and relevant papers cited there.}
% % 
% \marginXW{Added.}
% % 

These ideas and techniques are currently been extended to nonlinear
systems, see~\cite{TM-TBS-FA:23,CDP-PT:23} for recent comprehensive
surveys on data-driven control of nonlinear systems.  These
include~\cite{MG-CDP-PT:21,CDP-PT:21-ifac}, utilizing low-rank
approximation technique to handle system nonliterary and treat
reminders as disturbances; \cite{ZY-JC:22-csl}, studying data-driven
optimal control of bilinear systems; and \cite{AB-CDP-PT:20}, studying
the stabilization of bilinear systems and uses LMIs to create a
guaranteed region of attraction.  By generalizing Willems’ Lemma to
particular types of nonlinear systems,~\cite{JB-FA:20} provides
nonlinear data-driven control for Hammerstein-Wiener systems,
and~\cite{RS-JB-FA:21} for second-order discrete Volterra
systems. Based on dual stability theory and Farkas’
lemma,~\cite{TD-MS:20_2} develops data-driven control methods for
nonlinear continuous-time systems.
% %
% \marginJC{There is also \cite{ZY-JC:22-csl} for bilinear systems. A
%   bunch of recent papers cite ours, so some of those might also be
%   candidates for citing here. We don't want to have an infinite list,
%   so if there are some tutorial/overview papers, that'd be
%   enough. It's just that this data-driven control of nonlinear systems
%   area is exploding.}
% %
% \marginXW{Added.}
% %
Other methods for nonlinear data-driven control that are not based on
Willems’ Lemma include unfalsified control~\cite{MGS-TCT:95},
simultaneous perturbation stochastic approximation~\cite{JCS:92},
model-free adaptive control~\cite{ZH-SJ:11}, and iterative feedback
tuning~\cite{HH-SG-MG:94}. Different from the approaches described
above, where the controller can be directly computed using the data,
these approaches tune the parameters of the controller in an iterative
manner and gradually improve system performance. To address
stabilization to states other than the origin, methods like
\cite{CN-SF-SMS-MM:16} employ nonlinear basis functions and combine
model-inverse techniques with virtual reference feedback tuning for
reference tracking.
% Most existing data-driven control approaches focus on linear
% systems~\cite{JCW-PR-IM-BLMDM:05,CDP-PT:19,JC-JL-FD:19,JB-JK-MAM-FA:19}
% and stabilization to zero states. However, these methods are not
% directly applicable to nonlinear systems, such as linear-threshold
% networks. Furthermore, regulating system dynamics to non-zero
% states, which is a common requirement in practical applications, is
% straightforward for linear systems by shifting the system's origin
% but is non-trivial for nonlinear systems.

Here, we focus on developing data-driven control techniques for
linear-threshold network dynamics. These systems have diverse
applications across several domains.  In computational neuroscience,
linear-threshold networks have been used to model mesoscale brain
activity~\cite{PD-LFA:01}. Nodes of the system represent neuron
populations, with states indicating their average firing rates. Edge
weights capture excitatory or inhibitory interactions among physically
adjacent neurons, and a linear-threshold activation function accounts
for firing rate saturation due to hyperpolarization \cite{RBR-SAS:03}.
In social sciences, linear-threshold networks have been used to build
influence propagation models~\cite{WC-YY-LZ:10,YDZ-VS-NEL:17} to
characterize the dynamics of public opinions.  Nodes represent
individuals and their opinions, while edges describe interactions
influenced by others or social media (through external inputs to the
network). A linear-threshold is introduced to each node to gauge the
condition when individuals' opinions change.  In deep learning
applications with artificial neural networks, linear-threshold models
are the same as modified rectified linear units (RELU with
max-limits)~\cite{AGH-MZ-BC-DK-WW-TW-MA-HA:17}. RELUs with tunable
parameters are generally used within the hidden nodes in the deep
neural networks~\cite{VN-GEH:10}, which allows good robustness and
versatility for function approximation~\cite{DF-TS-AR:20}. Compared
with other common activation functions such as \texttt{sigmoid} and
\texttt{tanh}, RELU networks do not suffer vanishing gradient
problems~\cite{MML-KHL:17}, because its gradient is either a constant
or zero. For the same reason, ReLU networks have low computational
complexity~\cite{AK-IS-GEH:12} when performing gradient propagation.
One feasible approach to designing controllers for linear-threshold
networks is to first identify the system
parameters~\cite{XW-JC:25-tac} and then design a model-based
controller~\cite{EN-JC:21-tacI}. Instead, we pursue here a direct
data-driven approach that bypasses the system identification step to
avoid accumulating approximation errors.

% model identification for nonlinear
% systems~\cite{XH-RJM-SC-CJH-KL-GWI:08}

\subsubsection*{Statement of Contributions} 
We study the data-driven control of linear-threshold network dynamics
to stabilize the state of the system to a constant reference value. We
do this by designing two types of controllers: (i) a state feedback
controller with feed-forward reference input and (ii) an augmented
feedback controller with error integration. By utilizing sampled
input-output data of the system, we introduce a map that reconstructs
the state-input pair as a transformation of data matrices. This allows
us to obtain a purely data-based representation to describe the system
dynamics, avoiding an explicit identification of the parametric model.
We then use this result and the specific form of the state feedback
controller with feed-forward reference input to obtain a data-driven
representation of the closed-loop system. To design the feedback and
feedforward gain matrices,
% which bridges the stability of the system and our ultimate objective,
% i.e., the design itself of the feedback and feedforward gain matrices.
we view the resulting system as a switched system and formulate
conditions for the design of a common controller that stabilizes all
modes. The stabilization conditions are characterized by a set of
linear matrix inequalities (LMIs), whose solutions correspond to the
gain matrices. We provide a theoretical guarantee for the controller's
effectiveness in stabilizing the system state to the desired constant
reference value.  For non-zero reference states, however, we observe
that controllers with feed-forward inputs are often not robust to
system parameter changes and external disturbances. To address this
limitation, we design an augmented feedback controller with error
integration adapting the approach developed for the state feedback
controller to the new design. The process involves deriving a
closed-loop data-based representation and formulating stabilization
conditions in the form of LMIs.  In both cases, we observe that the
LMI formulations result in a number of equations that grow
exponentially with the system state's dimensionality. To address this
computational challenge, we introduce an alternative sufficient
condition that reduces the complexity of solving LMIs without
sacrificing performance.  Finally, we validate the effectiveness of
both algorithms in two case studies: regulating firing rate dynamics
in rodent brains and regulating arousal level dynamics in humans. The
results demonstrate that both controllers are effective and the
controller with error integration has an improved performance compared
with the one with feed-forward reference input for handling
disturbances.

\subsubsection*{Notation}
Let $\real$ denote the set of real numbers.  Let
${\bm 1}_r \in \real^r$ and ${\bm 1}_{r\times p} \in \real^{r\times p}$
denote the vector and matrix with all entries equal to
$1$, respectively. Let $I_r$ denote the $r\times r$ identity matrix.  We let
$\col\{A_1,A_2,\cdots,A_r\}=
\begin{bmatrix}
  A_1^{\top}
  &A_2^{\top}&
  \cdots & A_r^{\top}
\end{bmatrix}^{\top}$ be a vertical stack of matrices $A_1,\cdots,A_r$
possessing the same number of columns. We let
$\diag\{A_1,A_2,\cdots,A_r\}$ denote the diagonal stack of matrices
$A_1,\cdots,A_r$. We use  $\bm{x}[i]\in\mathbb{R}$ to denote  the
$i$th entry of 
vector $\bm{x}$; correspondingly,  $M[i,j]\in\mathbb{R}$ is the
entry of matrix $M$ on its $i$th row and $j$th column.  We denote by
$M^{\top}$ the transpose of~$M$. 
We let $\image(M)$ denote the linear span of the columns 
of matrix $M$. Specially, $\image(I_r)=\mathbb{R}^{r}$.
For symmetric matrices $M, N$, $M\succ (\succeq)N$ means
$M-N$ is positive (semi-) definite.
For $s>0$ and $x\in\mathbb{R}$, the
threshold function $[x]_0^s$ is defined as
$$
[x]_0^s=
\begin{cases}
	s & \text{for}\quad x>s ,
	\\
	x & \text{for}\quad 0\le x\le s ,
	\\
	0 & \text{for}\quad x< 0 ,
\end{cases}
$$
For a vector $\bm{x} \in \real^r$, $[\bm{x}]_0^s$ denotes the
component-wise application of this definition.

\section{Problem formulation}
Consider a linear-threshold network of $n$ nodes governed by the
following discrete-time dynamics:
\begin{align}\label{eq_dstmodel}
  \bm{x}(t+1)=\alpha \bm{x}(t) +
  \left[W\bm{x}(t)+B\bm{u}(t)\right]_0^{s},\quad t\in\mathbb{N}.
\end{align}
Here $\bm{x} \in \realnonneg^n$ is the network state, which is the
compact stack of a scalar state in each node.
$W\in\mathbb{R}^{n\times n }$ is the network connectivity matrix,
characterizing the strength of interactions between different nodes.
$\bm{u}(t)\in\mathbb{R}^m$ is the control input vector, and
$B\in\mathbb{R}^{n\times m}$ is the associated input matrix.  For each
node, the state of the node evolves according to an intrinsic decay
rate $\alpha\in(0,1)$, and the linear-threshold activation function
denoted by $[\cdot]_0^s$, whose input is co-generated by the node's
neighbors and external inputs.  Given the decay rate $\alpha$ and the
properties of linear-threshold functions, we assume the initial state
of the dynamics satisfies $0\le \bm{x}(0)\le \frac{s}{1-\alpha}$, so
that $0\le \bm{x}(t)\le \frac{s}{1-\alpha}$, for all $t\in\mathbb{N}$.
In~\eqref{eq_dstmodel}, we assume the parameters $\alpha$ and $s$ are
known, but the matrices $W$ and $B$ are unknown.  This type of network
dynamics arise in the modeling of the dynamical behavior of the firing
rates of neuronal populations~\cite{PD-LFA:01}, where $\alpha$ denotes
the intrinsic decay rate of neurons' firing behaviors, $W$ and $B$
characterize the inhibition/excitation response of neurons from other
connected neurons or external inputs, and the linear-threshold $s$
characterizes the firing rate saturation of neurons due to
hyperpolarization~\cite{RBR-SAS:03}. In such contexts, the assumption
above about known parameters and unknown matrices are reasonable.  The
dynamics also model the opinion propagation of individuals in social
networks~\cite{YDZ-VS-NEL:17}, where $\alpha$ denotes decaying
confidence \cite{IM-AG:10}, $W$ and $B$ characterize the change of
individuals' agreement/disagreement on the matter impacted by other
people or social media, and the linear-threshold characterizes the
saturation threshold of public opinion \cite{FG-SL-AM:08}. More
generally, model~\eqref{eq_dstmodel}, with tunable matrices $W$
and~$B$, can be used as artificial neural networks to approximate
nonlinear dynamics for learning and control
tasks~\cite{AGH-MZ-BC-DK-WW-TW-MA-HA:17,DF-TS-AR:20}.

To study dynamics~\eqref{eq_dstmodel}, we employ a data-driven
approach and assume that system inputs and states can be sampled from
experiments.  Let $T_d$ be the total number of available data points,
and let $\bm{x}_d^+(k)$, $\bm{x}_d(k)$, and $\bm{u}_d(k)$,
$k \in \until{T_d}$ denote the data samples (corresponding to
$\bm{x}(t+1)$, $\bm{x}(t)$, and $\bm{u}(t)$, respectively).  We employ
the index $k$ as an indicator that distinguishes one data sample from
another.  It is possible that all the sampling instances of the data
are chosen consecutively from a system trajectory, where all the data
samples are head-tail connected, i.e., $\bm{x}^+_d(k)$ of the former
data can be used as the $\bm{x}_d(k)$ of the latter one.  In general,
we allow data samples to be collected at independent time instances,
and even from various trajectories of the same system.

\begin{problem}\label{Prob}
  Consider system~\eqref{eq_dstmodel} with known parameters $\alpha$
  and $s$, and unknown matrices $W$ and~$B$.  Let
  $\bm{r}\in\mathbb{R}^n$, with $0\le \bm{r}\le \frac{s}{1-\alpha}$,
  be a constant reference value the desired state should converge to.
  Given data samples $\bm{x}_d(k)$, $\bm{x}^+_d(k)$ and $\bm{u}_d(k)$,
  $k \in \until{T_d}$, design a feedback controller which
  asymptotically stabilizes the system state $\bm{x}(t)$ to the
  reference value $\bm{r}$.
\end{problem}

In the context of the applications mentioned above, the data-driven
control Problem~\ref{Prob} is motivated by the following
considerations. In neuroscience, for instance, our brain continuously
regulates the firing rate of neuronal cells to specific patterns,
i.e., depending on different brain functions, the activity of certain
neurons should be excited or inhibited.  In social science, the
control problem arises when certain actors seek to steer public
opinion in a particular direction.
% by advocating the opinion of certain groups of people while
% depressing that of other groups.
In more general control applications where a linear-threshold
artificial neural network approximates certain nonlinear dynamics, the
manipulation of the latter can be achieved by controlling the neural
network
model.  %To stabilize the network dynamics~\eqref{eq_dstmodel},
% the controller we employ has a simple linear feedback form, which
% can stabilize the system regardless the operating mode of the
% linear-threshold model. Here, we see solving Problem~\ref{Prob} as a
% first step towards developing advanced techniques for the network
% dynamics~\eqref{eq_dstmodel}.  Including more sophisticated
% controller design, such as data-driven model predictive
% control~\cite{DP-MF-SF-AB:19}; and more complex control objectives,
% such as stabilizing the system to other equilibria (not necessarily
% the origin) and track various classes of dynamic trajectories.

%
%\marginJC{How about moving what we say in this paragraph to later in
  %the paper? Here, instead, we could jump straight to the paragraph
  %after this (will require removing refs there to (2), (3), and gain
  %matrices). Then, (i) the form (2) gets moved to Section 4, somewhere
  %in the beginning; (ii) The observations about the limitations of (2) can be
  %moved to somewhere in Section 4 (towards the end?), and (iii) the
  %form (3) gets moved to Section 5.}
%  \marginXW{Moved.}
%

We assume Problem \ref{Prob} is solvable. This means that: (i) the
matrices $W$ and $B$ are such that there exist controllers %in the
% forms \eqref{eq_theproblem} and \eqref{eq_intctrl}
that stabilize the
system \eqref{eq_dstmodel} to $\bm{r}$;
% \footnote{Reference \cite{EN-JC:21-tacI} considers a system
  %$\bm{x}(t+1)=\alpha \bm{x}(t)
  %+\left[W\bm{x}(t)+\mathtt{d}(t)\right]_0^{s}. $ For
  %\eqref{eq_theproblem}, one can associate $K_2\bm{r}$ with
  %$\mathtt{d}(t)$. For \eqref{eq_intctrl}, one can associate
  %$-K_1\bm{r}+ K_2\bm{\xi}(t)$ with $\mathtt{d}(t)$.}:
%
  %\marginJC{In the footnote, where did $K_1 x(t)$ goes? I'm tempted
  % to just remove the footnote.}
% \marginXW{Footnote removed. }
%
(ii) the data samples available are sufficiently rich to allow us to
design the controller %$K_1, K_2$
without knowing the matrices $W$ and~$B$.~\cite[Theorem
IV.8]{EN-JC:21-tacI} describes classes of linear-threshold systems for
which (i) holds. We provide below conditions that ensure that (ii)
holds.  In the following, we develop a data-driven approach (instead
of model-based) that directly synthesizes a controller to solve
Problem~\ref{Prob}.

\section{Data-Based Representation of Linear Threshold
  Models}\label{SEC_KI}

%
%\marginJC{How about calling ``Open-Loop Data-Based Representation''
%  simply ``Data-Based Representation''.  Later, we can keep 
%  ``Closed-Loop Data-Based Representation'' in titles, use also ``Data-Based
%  Representation of Closed-Loop System'' in text?}
%  \marginXW{Addressed!}
%

In this section, we provide a data-based reformulation of
the system~\eqref{eq_dstmodel}. This representation describes the
system update using the available data samples and does not require
knowledge of the unknown matrices $W$ and~$B$.

% \subsection{Open-loop data-based representation} \label{Sec_OLDR}
We start by defining $\bm{z}(t) = \bm{x}(t+1)-\alpha \bm{x}(t)$ and
$\bm{p}(t)=\col\{\bm{x}(t),\bm{u}(t)\}$. Then, the
update~\eqref{eq_dstmodel} can be rewritten as
\begin{align}\label{eq_dataF}
  \bm{z}(t)= \left[H \bm{p}(t)\right]_0^{s} ,
\end{align}
where
\begin{align*}%\label{eq_defzHP}
  H&=
     \begin{bmatrix}
       W&B
     \end{bmatrix}
     =\begin{bmatrix}[1.2]
        \horzbar 	& h_1^{\top} & \horzbar \\
        \horzbar 	& h_2^{\top} & \horzbar \\
                        &	\vdots &\\[.2em]
        \horzbar& 	h_n^{\top} & \horzbar
      \end{bmatrix}\in\mathbb{R}^{n\times (n+m)} .
\end{align*}
Let $\bm{h} = \col\{{h}_1,{h}_2,\dots,{h}_n\}\in\mathbb{R}^{n(n+m)}$
be a vectorized system parameter encoding the matrices $W$ and~$B$.
Then,
\begin{align*}%\label{eq_dataF2}  
  H \bm{p}(t)
  &=\begin{bmatrix}[1.2]
      h_1^{\top}\bm{p}(t)\\
      h_2^{\top}\bm{p}(t)\\
      \vdots\\
      h_n^{\top}\bm{p}(t)
    \end{bmatrix}
    =
    \begin{bmatrix}[1.2] {\bm{p}(t)}^{\top}{h}_1
      \\
      {\bm{p}(t)}^{\top}{h}_2
      \\
      \vdots
      \\
      {\bm{p}(t)}^{\top}{h}_n
    \end{bmatrix}=\left(I_{n}\otimes \bm{p}(t)^{\top}\right)\bm{h} .
\end{align*}
Thus, equation \eqref{eq_dataF} reads
\begin{align}\label{eq_dataF3}
  \bm{z}(t)= \left[\left(I_{n}\otimes
  \bm{p}(t)^{\top}\right)\bm{h}\right]_0^{s}.
\end{align}
In order to obtain a data-based representation for the system, a key
step is to represent $\bm{h}$ with the data samples
$\{\bm{x}_d^+(k), \bm{x}_d(k), \bm{u}_d(k)\}_{k=1}^{T_d}$.  Towards
this end, let
\begin{align*}%\label{eq_defXXP}
  \mathcal{Z}_d=
  \begin{bmatrix}[1.2] \bm{x}^+_d(1)-\alpha\bm{x}_d(1)
    \\
    \bm{x}^+_d(2)-\alpha\bm{x}_d(2)
    \\
    \vdots
    \\
    \bm{x}^+_d(T_d)-\alpha\bm{x}_d({T_d})
  \end{bmatrix}, 
  \;
  \mathcal{P}_d=
  \begin{bmatrix}[1.2]
    I_{n}\otimes \bm{p}_d^{\top}(1)
    \\
    I_{n}\otimes
    \bm{p}_d^{\top}(2)
    \\
    \vdots
    \\
    I_{n}\otimes \bm{p}_d^{\top}(T_d)
  \end{bmatrix} ,
\end{align*}
with $\bm{p}_d(k)=\col\{\bm{x}_d(k),\bm{u}_d(t)\}$,
$\mathcal{Z}_d\in\mathbb{R}^{nT_d}$ and
$\mathcal{P}_d \in\mathbb{R}^{nT_d\times n(n+m)}$. According
to~\eqref{eq_dataF3}, we have
\begin{align}\label{eq_compactdataeq}
  \mathcal{Z}_d=\left[\mathcal{P}_d\bm{h}\right]_0^{s}.
\end{align}
Define
$ f(\mathcal{Z}_d) =
\mathcal{P}_d\bm{h}-\left[\mathcal{P}_d\bm{h}\right]_0^{s} \in
\mathbb{R}^{nT_d}$ to represent the part of $\mathcal{P}_d\bm{h}$ that
is truncated by the linear threshold $[\cdot]_0^{s}$. Clearly,
$f(\mathcal{Z}_d)[i]\neq 0$ only if $\mathcal{Z}_d[i] = s$ or
$ \mathcal{Z}_d[i]=
0$. %$\mathcal{Z}_d[i] = \max(\mathcal{Z}_d) ~\text{or}~
%\mathcal{Z}_d[i]= 0$.
Thus, we define a diagonal matrix $E_d\in\mathbb{R}^{nT_d\times nT_d}$
such that for all $i \in \until{nT_d}$,
\begin{align}\label{eq_defcE}
  &E_d[i,i] =
    \begin{cases}
      \phantom{-}1 & \text{if } \mathcal{Z}_d[i] = s ~\text{or}~
                     \mathcal{Z}_d[i]= 0
      \\
      \phantom{-}0 & \text{otherwise.}
    \end{cases}
\end{align}
Then, there exists a vector $v \in \real^{nT_d}$ such that
\begin{align}\label{eq_defCv}
  f(\mathcal{Z}_d)=E_dv,
\end{align}
which, together with~\eqref{eq_compactdataeq}, yields $
\mathcal{P}_d\bm{h}=\mathcal{Z}_d +E_dv$.  Using the property that
$(I-E_d)E_d=E_d-E_d=0$, one has
\begin{align}\label{eq_eqcv}
  (I-E_d)\mathcal{P}_d\bm{h}=(I-E_d)\mathcal{Z}_d.
\end{align}
This equation describes the vectorized system parameter $\bm{h}$ in
the form of a data-based equality constraint.
% in the following, we show how it is related to the controller design
% formulated in Problem \ref{Prob}.
We make the following assumption on the data samples.

\begin{assumption}\longthmtitle{Data richness}\label{Ass_rankQ}
  Given data samples
  $\{ \bm{x}_d^+(k), \bm{x}_d(k), \bm{u}_d(k) \}_{k=1}^{T_d}$, the
  matrix $(I-E_d)\mathcal{P}_d$, has full column rank.
\end{assumption}

Assumption~\ref{Ass_rankQ} can be directly verified by
computation. Note that Assumption~\ref{Ass_rankQ} becomes easier to
satisfy as the number of data samples grows.  Based on this
assumption, we can combine equations \eqref{eq_dataF3}
and~\eqref{eq_eqcv} to obtain a data-based representation of the
system dynamics, as stated next.

\begin{lemma}\longthmtitle{Data-based
    representation}\label{LM_OLSYS} 
  % Given state-input pairs $\bm{p}(t)=\col\{\bm{x}(t),\bm{u}(t)\}$,
  % suppose there exists a mapping $F(\bm{p}(t)):
	% \mathbb{R}^{m+n}\to\mathbb{R}^{n\times nT_d}$ such that
  Under Assumption \ref{Ass_rankQ}, let $F:
  \mathbb{R}^{n+m}\to\mathbb{R}^{n\times nT_d}$ be such that
  \begin{align}\label{eq_consts}
    % F(\bm{p}(t))\cdot(I-E_d)\mathcal{P}_d=
    % I_{n}\otimes\bm{p}(t)^{\top} ,
    F(\bm{p})\cdot(I-E_d)\mathcal{P}_d= I_{n}\otimes\bm{p}^{\top},
  \end{align}
  for any state-input pair $\bm{p}=\col\{\bm{x},\bm{u}\}$.  Then the
  dynamics \eqref{eq_dstmodel} has the following data-based
  representation,
  \begin{align}\label{eq_DDMG}
    \bm{x}(t+1) =\alpha
    \bm{x}(t)+\left[F(\bm{p}(t))\cdot(I-E_d)\mathcal{Z}_d\right]_0^{s} .
  \end{align}
\end{lemma}
\begin{proof}
  Because $(I-E_d)\mathcal{P}_d$ has full column rank, by the
  Rouch\'{e}-Capelli Theorem~\cite{RAH-CRJ:12}, a map $F:
  \mathbb{R}^{m+n}\to\mathbb{R}^{n\times nT_d}$
  satisfying~\eqref{eq_consts} must exist (but may not be unique).
  Then, from equations \eqref{eq_dataF3}, \eqref{eq_eqcv} and
  \eqref{eq_consts}, it follows that
  \begin{align*}%\label{eq_DDMGPF}
    \bm{z}(t)
    &=\left[\left(I_{n}\otimes
      \bm{p}(t)^{\top}\right)\bm{h}\right]_0^{s}
    \\
    &=\left[F(\bm{p}(t))\cdot(I-E_d)\mathcal{P}_d\bm{h}\right]_0^{s}
    \\
    &=\left[F(\bm{p}(t))\cdot(I-E_d)\mathcal{Z}_d\right]_0^{s}.
  \end{align*}
  Equation~\eqref{eq_DDMG} follows since
  $\bm{x}(t+1)= \alpha \bm{x}(t) + \bm{z}(t) $.
\end{proof}

% For each $\bm{p}$, $F(\bm{p})$ transforms the data matrix
% $(I-E_d)\mathcal{P}_d$ into the augmented state-input pair
% $I_{n}\otimes\bm{p}^{\top}$ used by \eqref{eq_dataF3}. Based on this
% transformation, we obtain~\eqref{eq_DDMG}. 
We refer to~\eqref{eq_DDMG} as a data-based representation of
system~\eqref{eq_dstmodel} because it does not involve the matrices
$W$ and $B$ (or their vectorized version $\bm{h}$), and instead
allows, based on the data, to determine the state at the next timestep
based on the current state and the control input.  
The representation is valid for an arbitrary input.
%We refer to it as
%open-loop because the state-input pair is arbitrary and does not take
%into account the coupling of the system input and the system state
%characterized in controllers~\eqref{eq_theproblem} and
%\eqref{eq_intctrl}.  %We address this point next.
%
%\marginJC{If we change terminology, adjust last sentence. We could
%  simply say that the representation is valid for an arbitrary input.}
% \marginXW{Addressed!}
%

% In fact, if one explicitly obtains the function $F(\bm{p}(t))$ from
% the constraint \eqref{eq_consts}, then the process, in its nature,
% is identical to system
% identification. %Here, the identifiability depends on the rank of $(I-E_d)\mathcal{P}_d$.  In the following, we show how equation \eqref{eq_DDMG} can be further modified and utilized to achieve the controller design goal formulated in Problem \ref{Prob}.

\section{Data-Driven Control with Feed-forward Reference
  Input}\label{Sec_LControl}
In this section, we introduce a data-driven approach to solve
Problem~\ref{Prob} with a controller of the following
  form:
\begin{align}\label{eq_theproblem}
  \bm{u}(t)=K_1\bm{x}(t)+K_2\bm{r},
\end{align}
which is composed of a feed-back gain $K_1\in\mathbb{R}^{m\times n}$
over the system's current state and a feed-forward gain
$K_2\in\mathbb{R}^{m\times n}$ over the reference input.
%form of~\eqref{eq_theproblem}.
%
%\marginJC{This is where we could move to the form (2).}
%
Our strategy consists of first leveraging the data-based
representation obtained Section~\ref{SEC_KI} to describe the
closed-loop system.  We then view the resulting linear-threshold
system as a switched system and design a common linear feedback
controller that stabilizes all modes.

\subsection{Closed-loop data-based representation}\label{Sec_CLDR}
We start with a data-based representation of the closed-loop system 
dynamics under the controller~\eqref{eq_theproblem}.

\begin{lemma}\longthmtitle{Closed-loop data-based
    representation}\label{LM_CLSYS1}
  Let Assumption \ref{Ass_rankQ} hold.  The system \eqref{eq_dstmodel}
  under the controller~\eqref{eq_theproblem} admits the data-based
  representation,
  \begin{align}\label{eq_DDMGtr}
    \bm{x}(t+1)=\alpha
    \bm{x}(t)+\left[G(\bm{x}(t),\bm{r})\cdot(I-E_d)\mathcal{Z}_d\right]_0^{s},
  \end{align}
  where
  $G: \mathbb{R}^{n}\times \mathbb{R}^{n}\to\mathbb{R}^{n\times nT_d}$
  satisfies
  \begin{multline}\label{eq_conststr}
    G(\bm{x},\bm{r})\cdot(I-E_d)\mathcal{P}_d
    \\
    =
    I_{n}\otimes\left(\bm{x}^{\top}
    \begin{bmatrix}
      I_n & K_1^{\top}
    \end{bmatrix}
    +\bm{r}^{\top}
    \begin{bmatrix}
      \bm{0}_n & K_2^{\top}
    \end{bmatrix}\right) .
  \end{multline}
\end{lemma}
\smallskip
\begin{proof}	
  By defining
  \begin{align*}
    G(\bm{x},\bm{r})=F\left(
    \begin{bmatrix} 
      I_n\\K_1
    \end{bmatrix}
    \bm{x}+\begin{bmatrix} 
             \bm{0}_n\\K_2
           \end{bmatrix}\bm{r}\right) ,
  \end{align*}
  the result follows from Lemma~\ref{LM_OLSYS}.
\end{proof}

Equation~\eqref{eq_DDMGtr} provides a closed-loop data-based
description of system~\eqref{eq_dstmodel}.
% with the existence of $G$ following directly from the existence
% of~$F$.
%{\color{red} Note that the map $G$ cannot be directly designed due to
%  the implicit constraint~\eqref{eq_conststr}.}
%
%\marginJC{Do you rather mean something like: ``A closed-form
%  expression for the map $G$ is not readily available from
%  equation~\eqref{eq_conststr}.''  }
%
Nevertheless, a closed-form expression for the map $G$ is not readily
available from equation~\eqref{eq_conststr}.  In what follows, we
provide an explicit construction of this map.

Looking at the right-hand side of~\eqref{eq_conststr}, we observe a
block-diagonal matrix resulting from the Kronecker product.  We make
use of such special pattern as follows.  Define
\begin{align*}
  \overline{\mathcal{P}}_d=I_n\otimes
  \begin{bmatrix}[1.2]
    \bm{p}_d^{\top}(1)\\
    \bm{p}_d^{\top}(2)\\
    \vdots\\
    \bm{p}_d^{\top}(T_d)
  \end{bmatrix} \in\mathbb{R}^{nT_d\times n(n+m)}.
\end{align*}
Because $\overline{\mathcal{P}}_d$ and ${\mathcal{P}}_d$ share
  the same rows but in different orders, there  exist a permutation matrix
  $T_F\in \mathbb{R}^{nT_d\times nT_d}$ such that 
  % and let $T_F\in \mathbb{R}^{nT_d\times nT_d}$ be the elementary
  % matrix that switches the rows of ${\mathcal{P}}_d$ to obtain
  % $\overline{\mathcal{P}}_d$, i.e.,
%
  % \marginJC{Why is it elementary?}  \marginXW{Rewritten. It should
  % be a permutation matrix that encodes multiple row
  % switching. Instead of elementary.}
%
\begin{align}\label{eq_defbPd}
  \overline{\mathcal{P}}_d =T_F{\mathcal{P}}_d.
\end{align}
Since $E_d$ is a diagonal matrix, $\overline{E}_d=T_F E_d T_F^{-1}\in
\mathbb{R}^{nT_d\times nT_d}$
% 
%\marginJC{The inverse of an elementary matrix is itself, no?}
% \marginXW{$T_F^{-1}=T_F^{\top}$}
%
is also a diagonal matrix\footnote{Let the permutation matrix $T_F$
  correspond to the permutation tuple $\pi$, which is a reordering of
  the set $\{1,2,\cdots,nT_d\}$. By definition of a permutation
  matrix, one has $T_F^{-1}=T_F^{\top}$, and therefore
  $\overline{E}_d[j,k]=E_d[\pi[j],\pi[k]]$~\cite{HW:14}. Consequently,
  $\overline{E}_d[j,k]=0$ if $j\neq k$.}
%
% \marginJC{Why? Reference? Is this true for any diagonal matrix, or
% just true b/c $E_d$ is a particular type of diagonal matrix (with
% 1's or 0's)?}
% \marginXW{Is it worth a lemma? I currently put it as a footnote.}
%
which we write as
\begin{align*}
  \overline{E}_d=
  \diag\{\overline{\mathtt{E}}_1,\cdots,\overline{\mathtt{E}}_n\} ,
  % \begin{bmatrix}[1.2]
  % \overline{\mathtt{E}}_1 &&\\
  % &\ddots&\\
  % &&\overline{\mathtt{E}}_n
  % \end{bmatrix} ,
\end{align*}
with $\overline{\mathtt{E}}_i\in\mathbb{R}^{T_d\times T_d}$.  Now, let
\begin{align*}
  \begin{bmatrix}[1.2]
    \overline{z}_1 \\
    \vdots\\
    \overline{z}_n
  \end{bmatrix} = T_F \mathcal{Z}_d ,
\end{align*}
with $\overline{z}_i\in\mathbb{R}^{T_d}$, $i\in\until{n}$, and define
$Z=\diag\{Z_1,\cdots,Z_n\}$ by
\begin{subequations}\label{eq:new-data-matrices}
  % \begin{align}\label{eq_defZ}
  %   Z=\begin{bmatrix} \overline{z}_1^{\top}(I-\overline{E}_1) &&
  %     \\
  %     &\ddots&
  %     \\
  %     &&\overline{z}_n^{\top}(I-\overline{E}_n)
  %   \end{bmatrix}\in\mathbb{R}^{\times T_d}.
  % \end{align}
  \begin{align}\label{eq_defZ}
    Z_i=
    \overline{z}_i^{\top}(I_{T_d}-\overline{\mathtt{E}}_i)\in\mathbb{R}^{1\times
    T_d}. 
  \end{align}
  %
  %\marginJC{This $I$ is $I_{T_d}$, no? Here and just below in the def
  %  of $Q_i$}
  %
  Further define
  $Q=\diag\{Q_1,\cdots,Q_n\} \in\mathbb{R}^{nT_d\times n(n+m)}$ by
  \begin{align}\label{eq_defQ}
    Q_i=\left(I_{T_d}-\overline{\mathtt{E}}_i\right)
    \begin{bmatrix}[1.2]
      \bm{p}_d^{\top}(1)
      \\
      \bm{p}_d^{\top}(2)
      \\
      \vdots
      \\
      \bm{p}_d^{\top}(T_d)
    \end{bmatrix} \in\mathbb{R}^{T_d\times (n+m)}.
  \end{align}  
\end{subequations}
By definition, we have $Q=(I_{nT_d}-\overline{E}_d)\overline{\mathcal{P}}_d$.
%
%\marginJC{And this $I$ is $I_{nT_d}$}
%

\begin{lemma}\longthmtitle{Modified closed-loop data-based
    representation}\label{LM_CLSYS}
  Let Assumption~\ref{Ass_rankQ} hold and consider the data matrices
  $Z$ and $Q$ defined in~\eqref{eq:new-data-matrices}. The
  system~\eqref{eq_dstmodel} under the
  controller~\eqref{eq_theproblem} can be represented by
  \begin{align}\label{eq_DDMZMtr}
    \bm{x}(t+1) = \alpha \bm{x}(t) + \left[Z(M \bm{x}(t)+N \bm{r})\right]_0^{s},
  \end{align}
  where $M, N\in\mathbb{R}^{nT_d\times n}$ satisfy
  \begin{align}\label{eq_QMK}
    Q^{\top}M=\bm{1}_n\otimes
    \begin{bmatrix} I_n\\K_1
    \end{bmatrix},
    \quad Q^{\top}N =
    \bm{1}_n\otimes
    \begin{bmatrix}
      \bm{0}_{n \times n}\\K_2
    \end{bmatrix}.
  \end{align}
  %
  %\marginJC{With the way we've defined the notation, $ \bm{0}_n$ is a
  %  vector of $n$ zero's, but for the dimensions to match, shouldn't
  %  it be  $\bm{0}_{n \times n}$?}
%   \marginXW{Addressed!}
  %
\end{lemma}
\begin{proof}
  We first establish the existence of~$M$, $N$. Note that
  \begin{align*}
    Q&=(I-\overline{E}_d)\overline{\mathcal{P}}_d =(T_FT_F^{-1}-T_F
       E_d T_F^{-1})T_F{\mathcal{P}}_d
    \\
     &=T_F(I- E_d ){\mathcal{P}}_d .
  \end{align*}
  Since $(I-E_d)\mathcal{P}_d$ has full column rank, cf.  Assumption
  \ref{Ass_rankQ}, we deduce that $Q^{\top}$ must have full row
  rank. From the Rouch\'{e}-Capelli Theorem~\cite{RAH-CRJ:12}, for any
  $K_1, K_2$, there always exist $M, N$ (which may not be unique)
  satisfying equation~\eqref{eq_QMK}.
  
  Next, we show that $M, N$ can be used to construct the map $G$
  satisfying~\eqref{eq_conststr}. Let $M=\col\{M_1,\cdots,M_n\}$,
  $N=\col\{N_1,\cdots,N_n\}$, with
  $M_i, N_i\in\mathbb{R}^{T_d\times n}$.  From the diagonal structure
  of $Q$, \eqref{eq_QMK} can be equivalently rewritten as
  \begin{align}\label{eq_QMIK}
    Q_i^{\top}M_i
    =\begin{bmatrix}
       I_n\\K_1
     \end{bmatrix} ,\quad 
    Q_i^{\top}N_i
    =\begin{bmatrix}
       \bm{0}_{n\times n}\\K_2
     \end{bmatrix},
  \end{align}
  for $i\in\until{n}$. Let
  $\overline{G}_i(\bm{x},\bm{r})=(M_i\bm{x}+N_i\bm{r})^{\top} \in
  \real^{T_d}$, $i\in\until{n}$. Note that
  \begin{align*}%\label{eq_consts31}
    \overline{G}_i(\bm{x},\bm{r})Q_i
    &=
      \bm{x}^{\top}M_i^{\top}Q_i+\bm{r}^{\top}N_i^{\top}Q_i
    \\
    &=\bm{x}^{\top}
      \begin{bmatrix}
        I_n & K_1^{\top}
      \end{bmatrix}
      +\bm{r}^{\top}
      \begin{bmatrix}
        \bm{0}_{n\times n} & K_2^{\top}
      \end{bmatrix} .
  \end{align*}
  Define
  $\overline{G}(\bm{x},\bm{r}) = \diag\{\overline{G}_1
  (\bm{x},\bm{r}),\cdots,\overline{G}_n(\bm{x},\bm{r})\}$.  One then has
  \begin{multline}\label{eq_conststrT}
    \overline{G}(\bm{x},\bm{r})(I-\overline{E}_d)\overline{\mathcal{P}}_d
    \\
    =
    I_{n}\otimes\left(\bm{x}^{\top}
      \begin{bmatrix}
        I_n & K_1^{\top}
      \end{bmatrix}
      +\bm{r}^{\top}
      \begin{bmatrix}
        \bm{0}_{n\times n} & K_2^{\top}
      \end{bmatrix}\right).
  \end{multline}
  By letting $G(\bm{x},\bm{r})=\overline{G}(\bm{x},\bm{r})T_F$, it
  follows that 
  \begin{align*}%\label{eq_bGP=GP}
    G(\bm{x},\bm{r})(I-E_d)\mathcal{P}_d
    &=\overline{G}(\bm{x},\bm{r})T_F
      T_F^{-1}(I-\overline{E}_d) % T_F T_F^{-1}
      \overline{\mathcal{P}}_d
    \\
    &=\overline{G}(\bm{x},\bm{r})(I-\overline{E}_d)\overline{\mathcal{P}}_d .
  \end{align*}
  This, together with~\eqref{eq_conststrT}, implies that $G$
  satisfies~\eqref{eq_conststr}.
  
  Finally, we show that~\eqref{eq_DDMZMtr} follows from this fact and
  equation~\eqref{eq_DDMGtr}, as follows
  \begin{align*}
    & \left[Z(M \bm{x}+N\bm{r})\right]_0^{s}=
      \begin{bmatrix}[1.2]
        \overline{z}_1^{\top}(I-\overline{\mathtt{E}}_1)(M_1\bm{x}+N_1\bm{r})
        \\
        \vdots
        \\
        \overline{z}_n^{\top}(I-\overline{\mathtt{E}}_n)(M_n\bm{x}+N_n\bm{r})
      \end{bmatrix}_0^{s}
    \\
    & =
      \begin{bmatrix}[1.2]
        \overline{z}_1^{\top}
        (I-\overline{\mathtt{E}}_1)(\overline{G}_1(\bm{x},\bm{r}))^{\top} 
        \\
        \vdots\\
        \overline{z}_n^{\top}
        (I-\overline{\mathtt{E}}_n)(\overline{G}_n(\bm{x},\bm{r}))^{\top}  
      \end{bmatrix}_0^{s}
      \!=
      \begin{bmatrix}[1.2]
        \overline{G}_1(\bm{x},\bm{r})(I-\overline{\mathtt{E}}_1)\overline{z}_1
        \\
        \vdots\\
        \overline{G}_n(\bm{x},\bm{r})(I-\overline{\mathtt{E}}_n)\overline{z}_n
      \end{bmatrix}_0^{s}
    \\
    &
      =
      \left[\overline{G}(\bm{x},\bm{r})
      (I-\overline{E}_d)T_F\mathcal{Z}_d\right]_0^{s}   
      =\left[G(\bm{x},\bm{r})(I-E_d)\mathcal{Z}_d\right]_0^{s} ,
  \end{align*}
  which completes the proof.
\end{proof}

Comparing the statements in Lemmas~\ref{LM_CLSYS1} and~\ref{LM_CLSYS},
we see that finding
$G: \mathbb{R}^{n}\times \mathbb{R}^{n}\to\mathbb{R}^{n\times nT_d}$
satisfying~\eqref{eq_conststr} can be accomplished by finding the
matrices $M, N\in\mathbb{R}^{nT_d\times n}$ satisfying~\eqref{eq_QMK}.
Regarding the solution of this latter equation, we note that when
$K_1, K_2$ are given, $M, N$ can be readily computed. In fact,
equations~\eqref{eq_QMIK} prescribe exactly how to find
$\{M_i, N_i\}_{i=1}^n$.  However, since our ultimate goal is to design
controller gain matrices $K_1, K_2$ themselves, the solution in $M, N$
and $K_1, K_2$ of equation~\eqref{eq_QMK} poses the challenge of
jointly solving for $n$ coupled systems of linear equations.  The
following result provides a reformulation of the equation showing that
$K_1, K_2$ can be expressed as functions of~$M, N$.

\begin{lemma}\longthmtitle{Decoupling constraints for closed-loop
    data-based representation}\label{LM_decouple}
  Define matrices $L\in\mathbb{R}^{n\times n}$,
  $C_1 \in \mathbb{R}^{n\times n(n+m)}$ and $C_2\in\mathbb{R}^{m\times
    n(n+m)}$ as
  \begin{align*}
    L[i,j]
    &=
      \begin{cases}
        n-1 & \text{if } i=j
        \\
        -1 & \text{otherwise} ,
      \end{cases}
    \\
    C_1
    &=
      \begin{bmatrix}
        I_{n} & \bm{0}_{n \times n(n+m)-n}
      \end{bmatrix},
    \\
    C_2
    &=
      \begin{bmatrix}
        \bm{0}_{m \times n}
        & I_{m}
        &
          \bm{0}_{m \times (n-1)(n+m)}
      \end{bmatrix},
  \end{align*}
  and let
  $\mathcal{L}=L\otimes I_{n+m} \in \mathbb{R}^{n(n+m)\times n(n+m)}
  $. Then, equation~\eqref{eq_QMK} can be equivalently written as
  \begin{subequations}\label{eq_matrix_constraints}
    \begin{align}%{2}
      \mathcal{L}Q^{\top}M
      &=
        \bm{0}_{n(n+m)\times n} ,
      &
        \mathcal{L}Q^{\top}N
      & =\bm{0}_{n(n+m)\times n} ,
        \label{eq_LQMLQN} 
      \\
      C_1Q^{\top}M
      &=I_n,
      &
        C_1Q^{\top}N
      &
        =\bm{0}_{n \times n} ,
        \label{eq_LQMLQN2}
      \\
      C_2 Q^{\top}M
      &=K_1,
      &
        C_2 Q^{\top}N & = K_2.\label{eq_KQM}
    \end{align}
  \end{subequations}
\end{lemma}
\begin{proof}
  From the definition of $L$, one has $\ker (L)= \image
  (\bm{1}_n)$.
  %
  % \marginJC{Shouldn't the diagonal entries of $L$ be $n-1$ instead
  % of $n$ for this to be true? Also, by $\image$ you mean linear
  % span?  We did not define it in the notation}
  % \marginXW{Sorry for the mistake here. Yes, the diagonal entries
  %   should be $n-1$. I should have defined $\image$. An image is the
  %   same as a range (I don't know which one is more commonly used),
  %   which is the linear span of the columns of a matrix. I was using
  %   $\image$ because I thought it was a counterpart concept of $\ker$,
  %   but probably it is not. By definition, one has
  %   $\image(I_{n+m})= \mathbb{R}^{n+m}$. Then
  %   $\ker(\mathcal{L}) = \bm{1}_n \otimes\mathbb{R}^{n+m}=
  %   \bm{1}_n\otimes\image(I_{n+m})= \image(\bm{1}_n\otimes I_{n+m})$
  %   The two representations mean the same thing.  I actually think
  %   your version is more concise. But for
  %   $\ker (L)= \image(\bm{1}_n)$, probably it is better than
  %   $\ker (L)= \bm{1}_n \cdot \mathbb{R}$? }
  %
  Thus, $\ker(\mathcal{L}) = \bm{1}_n\otimes\image(I_{n+m})$.
  % 
  % \marginJC{I don't think this is true (for instance, the dims do not
  % match). I think you mean to say
  % $\ker(\mathcal{L}) = \bm{1}_n \otimes
  % \mathbb{R}^{n+m}$, no? That way
  % $L \otimes I_{n+m} (\bm{1}_n \otimes z) = L \bm{1}_n \otimes z = 0 $ }
  % 
  From \eqref{eq_LQMLQN}, for $i,j\in\until{n}$, it holds
  \begin{align}\label{eq_QMQM}
    Q_i^{\top}M_i = Q_j^{\top}M_j,\quad Q_i^{\top}N_i = Q_j^{\top}N_j .
  \end{align}
  Furthermore, from \eqref{eq_LQMLQN2} and \eqref{eq_KQM}, one has 
  \begin{align}\label{eq_QMIK2}
    Q_1^{\top}M_1=
    \begin{bmatrix}
      I_n\\K_1
    \end{bmatrix} ,
    \quad Q_1^{\top}N_1=
    \begin{bmatrix}
      \bm{0}_{n\times n}\\K_2
    \end{bmatrix}.
  \end{align}
  Thus, equations \eqref{eq_QMQM} and \eqref{eq_QMIK2} are equivalent
  to \eqref{eq_QMIK} for $i\in\until{n}$, which corresponds
  to~\eqref{eq_QMK}. Equation~\eqref{eq_KQM} readily follows by left
  multiplying~\eqref{eq_QMK} by $C_2$.
\end{proof}

The advantage of the constraint
formulation~\eqref{eq_matrix_constraints} over~\eqref{eq_QMK} is that,
in the former, $K_1, K_2$ can be expressed as functions of~$M,
N$, respectively.  In fact, as we vary $M,
N$ among all possible solutions
of~\eqref{eq_LQMLQN}-\eqref{eq_LQMLQN2}, $K_1,
K_2$ in~\eqref{eq_KQM} take every possible value in
$\mathbb{R}^{m\times
  n}$.  This means that we can use~\eqref{eq_DDMZMtr}
and~\eqref{eq_LQMLQN}-\eqref{eq_LQMLQN2} to design the closed-loop
behavior of the system, with $M,
N$ as the only variables. Then, to implement the desired closed-loop
system, one can compute $K_1,
K_2$ using \eqref{eq_KQM} and apply it to the controller
\eqref{eq_theproblem} of the closed-loop system.  In the
equations~\eqref{eq_DDMZMtr} and~\eqref{eq_LQMLQN}-\eqref{eq_LQMLQN2},
note that $\mathcal{L}$, $C_1$, $C_2$ are constant matrices; and
$Z$ and $Q$ are matrices constructed from the data samples.

\subsection{LMI-based design of feedback gain
  matrices}\label{Subsec_feedback_LMI}
Based on the closed-loop data-based representations obtained in
Section~\ref{Sec_CLDR}, we now introduce a data-driven approach to
design the gain matrices for~\eqref{eq_theproblem}.  Note that the
system~\eqref{eq_dstmodel} is nonlinear due to the presence of the
threshold function.  Our strategy is to view it as a switched
system and design a common controller that applies to all modes.

We start by defining an error term
$\bm{\epsilon}(t)=\bm{x}(t)-\bm{r}$.  By subtracting $\bm{r}$ on both
sides of \eqref{eq_DDMZMtr}, one has
\begin{align}\label{eq_errmodel}
  \bm{\epsilon}(t+1)
  &= \alpha \bm{\epsilon}(t) + \left[Z(M \bm{x} (t)+N
    \bm{r}) \right]_0^{s}-(1-\alpha)\bm{r}. 
\end{align}
To ensure $\bm{x}(t)=\bm{r}$, i.e., $\bm{\epsilon}(t)=0$, is the
equilibrium of the system, one needs
$Z(M \bm{r} +N \bm{r}) =(1-\alpha)\bm{r}$. Since the reference signal
might be anything satisfying $0\le\bm{r}\le\frac{s}{1-\alpha}$, this
necessitates
\begin{align}\label{eq_ZMNC} 
  Z(M+N)&=(1-\alpha)I_n .
\end{align}
Based on \eqref{eq_ZMNC}, we rewrite \eqref{eq_errmodel} as
\begin{align}\label{eq_dstmodels0} 
  \bm{\epsilon}(t+1)
  &=\alpha \bm{\epsilon}(t) + \left[ZM
    \bm{\epsilon}(t)+(1-\alpha)\bm{r}\right]_0^{s} -
    (1-\alpha)\bm{r}\nonumber
  \\
  &=\alpha \bm{\epsilon}(t) + \overline{R}(\bm{\epsilon}(t),\bm{r})ZM
    \bm{\epsilon}(t) ,
\end{align}
where $\overline{R}(\bm{\epsilon},\bm{r})\in\mathbb{R}^{n\times n}$ is
a diagonal matrix with each entry
$\overline{R}(\bm{\epsilon},\bm{r})[i,i]$ defined as
\begin{align*}
  \begin{cases}
    \frac{(\left[ZM \bm{\epsilon}+(1-\alpha)\bm{r}\right]_0^{s} -
    (1-\alpha)\bm{r})[i]}{(ZM \bm{\epsilon})[i]}
    &\text{if {\small$(ZM
      \bm{\epsilon}+(1-\alpha)\bm{r})[i]>s$}}
    \\
    &~~\text{or {\small$(ZM \bm{\epsilon}+(1-\alpha)\bm{r})[i]<0$}} ,
    \\ 
    1 &\text{otherwise} .
  \end{cases}
\end{align*}
This represents equation~\eqref{eq_dstmodels0} as a switched system,
where the matrix $\overline{R}(\bm{\epsilon},\bm{r})$ depends on the
state.  This dependency makes challenging the stability analysis. To
address this, we perform an overapproximation by constructing a convex
hull that encompasses all possible values of
$\overline{R}(\bm{\epsilon},\bm{r})$.

By definition, one can derive that, for $i\in\until{n}$,
\begin{align}\label{eq_rangeR}
    0< \overline{R}(\bm{\epsilon},\bm{r})[i,i]\le 1.
\end{align}
Here $\overline{R}(\bm{\epsilon},\bm{r})[i,i]\neq0$ because the
equality holds only if
$(\left[ZM \bm{\epsilon}+(1-\alpha)\bm{r}\right]_0^{s} -
(1-\alpha)\bm{r})[i]=0$, which means $(ZM \bm{\epsilon})[i]=0$,
contradicting  the condition that
$(ZM\bm{\epsilon}+(1-\alpha)\bm{r})[i]>s$ or
$(ZM \bm{\epsilon}+(1-\alpha)\bm{r})[i]<0$.

To continue, define diagonal matrices $R_j \in \real^{n \times n}$,
$j\in\until{2^n}$, covering all possibilities such that
$R_j[i,i]\in\{0,1\}$ (note that this means that one of these matrices
is $\bm{0}_{n\times n}$).
%
% \marginJC{One of these is the zero matrix, right?}
% \marginXW{Yes, one of these is a zero matrix, and it is 
% needed. 
% For example if $\bar{R}=\begin{bmatrix}
%     0.2& 0\\ 0&0.3 
% \end{bmatrix}$. Then
% $\bar{R}=0.2 R_1+ 0.1 R_3 +0.7R_4$,
% where $R_1=I$, $R_3=\begin{bmatrix}
%     0& 0\\ 0&1 
% \end{bmatrix}$, $R_4=0$.
% We can only guarantee there exists a set of 
% coefficients such that $\mu_1>0$,
% but we cannot guarantee there exists a set of coefficients such that 
% $\mu_4=0$ (for the zero matrix $R_j$),
% even though the all diagonal entries of $\bar{R}$ are positive.
% }
% \marginXW{I have changed the index from $i$ to $j\in\until{2^n}$
% throughout the paper. When using $i$, it has a limit of $n$,
% when using $j$, it has a limit of $2_n$. I hope this makes the paper
% easier to follow.
% }
%
We let $R_1=I_n$. Since
$\overline{R}(\bm{\epsilon},\bm{r})$ is a diagonal matrix satisfying
\eqref{eq_rangeR}, regardless of the values of $\bm{\epsilon},\bm{r}$,
it can always be represented as a convex combination:
\begin{align}\label{eq_Rhull}
    \overline{R}(\bm{\epsilon},\bm{r})= \sum_{j=1}^{2^n}\mu_j
  R_j,~~\text{with}~~\sum_{j=1}^{2^n}\mu_j =1 
    ~~\text{and}~~\mu_1>0.
\end{align}
Here, $\mu_1>0$ because $R_1=I_n$ and all diagonal entries of
$\overline{R}(\bm{\epsilon},\bm{r})$ are strictly positive.

Based on equation \eqref{eq_Rhull}, the following result provides a way 
to synthesize the controller.

% To represent all possible models, we find it convenient to define
% $R_i \in \real^{n \times n}$,
% $i\in\until{2^n}$ %as the vertices of the unit hypercube $[0,1]^{n}$, such that $R_i[j,j]\in\{0,1\}$,
%
% \marginJC{I'm confused, $R_i$ is a matrix, what does $R_i[j]$ mean?}
% \marginXW{Changed to $R_i[j,j]$}
%
%$j\in\until{n}$.

%
%\marginJC{How can a matrix be a vertex for a hypercube? Is there some
% identification here going on that we have not made explicit?}
% \marginXW{Removed. I agree it is confusing.}
%
%
%\marginJC{What is the connection between $R_i$'s and
%  $\overline{R}(\bm{\epsilon},\bm{r})$?}
% \marginXW{Added.}
%
% Define
% $$\mathcal{R}=\text{Cov}(R_i) \quad\text{for}\quad
% i\in\until{2^n}.$$
% Based on $R_i$, we have the following result.

% \marginXW{I changed the condition in the theorem to a
% slightly relaxed version, by levering the property
% that $\overline{R}(\bm{\epsilon},\bm{r})$ has
% strictly positive diagonal entries.}
\begin{theorem}\longthmtitle{Data-driven
    synthesis via LMIs}\label{Thm_stable}
  Let the matrices $\overline{P}\in\mathbb{R}^{n\times n}$ and
  $S_1,S_2\in\mathbb{R}^{nT_d\times n}$ satisfy
  \begin{subequations}\label{eq_LMIs}
    \begin{align}
    \begin{bmatrix}\label{eq_LMIstable0}
        \overline{P}
        &
          ~~\left(\alpha \overline{P} + 
          ZS_1\right)^{\top} 
        \\
        \alpha \overline{P} + ZS_1 & \overline{P}
      \end{bmatrix}
      &\succ 0
        ,
      \\
      \begin{bmatrix}
        \overline{P}
        &
          ~~\left(\alpha \overline{P} + R_j
          ZS_1\right)^{\top} \label{eq_LMIstable}
        \\
        \alpha \overline{P} + R_j ZS_1 & \overline{P}
      \end{bmatrix}
      & \succeq 0 ,
      \\    
      \mathcal{L}Q^{\top}S_1= \bm{0}_{n(n+m)\times n}, 
      \quad
      ~~~\mathcal{L}Q^{\top}S_2
        &=\bm{0}_{n(n+m)\times n}
          , \label{eq_LCstable1} 
      \\
      C_1Q^{\top}S_1=\overline{P}, ~~~~~~~~~~~ \quad
      ~C_1Q^{\top}S_2
        &= \bm{0}_{n \times n}, \label{eq_LCstable2}
      \\
      Z(S_1+S_2)=(1-\alpha)\overline{P} , \label{eq_LCstable3} 
    \end{align}
  \end{subequations}
  for $R_j$, $j\in\{2,3,\dots,2^n\}$.  Then, the
  controller~\eqref{eq_theproblem},
  % $\bm{u}(t)=K_1\bm{x}(t)+K_2\bm{r}$,
  with
  \begin{align}\label{eq_KQS}
    \begin{aligned}
      K_1=C_2 Q^{\top}S_1 \overline{P}^{-1},\quad  K_2=C_2 Q^{\top}S_2
      \overline{P}^{-1},
    \end{aligned}
  \end{align}
  ensures that $\bm{r}$ is asymptotically stable for 
  the closed-loop system.
  %drives the state of system \eqref{eq_dstmodel} to converge
  %asymptotically to $\bm{r}$.
  %
  %\marginJC{Wording is a bit odd: ``drive state .. to converge
  %  asymptotically''. Isn't it better to say that $\bm{r}$ is
  %  asymptotically stable for closed-loop system?}
%  \marginXW{Addressed.}
  %
\end{theorem}
\begin{proof}
  To prove the statement, we construct a quadratic Lyapunov function
  of the form $V(\bm{\epsilon}) = \bm{\epsilon}^{\top}P\bm{\epsilon}$,
  with $P\succ0$. The convergence result holds if for any
  $\bm{\epsilon}(t)\neq 0$, the function satisfies
  \begin{align*}
    V(\bm{\epsilon}(t+1))<V(\bm{\epsilon}(t)).
  \end{align*}
  From \eqref{eq_dstmodels0}, this requires
  \begin{align}
    \bm{\epsilon}
    &\!^{\top}\!\!\left(\left(\alpha I\! +\!
      \overline{R}(\bm{\epsilon},\bm{r}) ZM\right)\!^{\top}\!P\!\left(\alpha I
      \!+\! \overline{R}(\bm{\epsilon},\bm{r}) ZM\right)\right)\bm{\epsilon}
      <\bm{\epsilon}P\bm{\epsilon} ,
      \label{eq_Lypieq}
  \end{align}
  for all $\bm{\epsilon}\neq 0$. 
  To make sure the above equation
  holds, a typical and sufficient approach in switched 
  systems~\cite{HL-PJA:09} is to guarantee the following matrix 
  inequality holds 
  \begin{align}\label{eq_lmi0}
    \left(\alpha I + \overline{R}(\bm{\epsilon},\bm{r})
    ZM\right)^{\top}P\left(\alpha I + \overline{R}(\bm{\epsilon},\bm{r})
    ZM\right)- P \prec0,
  \end{align}
  for all possible values of $\overline{R}(\bm{\epsilon},\bm{r})$.
  This effectively bypasses the dependency of
  $\overline{R}(\bm{\epsilon},\bm{r})$ on~$\bm{\epsilon}$.  Using the
  Schur complement~\cite{CKL-RM:04},
  % 
  % \marginJC{Reference here for Schur complement.}  \marginXW{Added.}
  % 
  \eqref{eq_lmi0} is equivalent to
  \begin{align}\label{eq_lmi1}
    \begin{bmatrix}
      {P} & \left(\alpha I + \overline{R}(\bm{\epsilon},\bm{r})
            ZM\right)^{\top}
      \\
      \alpha I + \overline{R}(\bm{\epsilon},\bm{r}) ZM & {P}^{-1}
    \end{bmatrix}\succ 0.
  \end{align}
  Based on the convex combination in~\eqref{eq_Rhull}, it is sufficient to 
  consider the following condition
  \begin{subequations}\label{eq_lmir}
    \begin{align}\label{eq_lmira}
      \begin{bmatrix}
        {P} & \left(\alpha I + 
              ZM\right)^{\top}
        \\
        \alpha I + ZM & {P}^{-1}
      \end{bmatrix}\succ 0,
      \\
      \begin{bmatrix}
        {P} & \left(\alpha I + R_j
              ZM\right)^{\top}
        \\
        \alpha I + R_jZM & {P}^{-1}
      \end{bmatrix}\succeq 0.
    \end{align}
  \end{subequations}
  for a common $P\succ0$, for all $j\in\{2,3,\dots,2^n\}$. Here
  \eqref{eq_lmira} is obtained by letting $R_1=I_n$. Since in
  \eqref{eq_Rhull}, $\mu_1>0$, the convex combination of
  \eqref{eq_lmir} guarantees \eqref{eq_lmi1}.
  
  Now, let $\overline{P}=P^{-1}$.  Since $\overline{P}\succ 0$,
  without losing generality, we introduce two new matrices
  $S_1\triangleq M\overline{P}$ and $S_2\triangleq N\overline{P}$.
  Pre- and post-multiplying equation \eqref{eq_lmir} by
  $\begin{bmatrix}
     \overline{P}
     &
       0
     \\
     0
     &
       I
   \end{bmatrix}$
   yields conditions \eqref{eq_LMIstable0}-\eqref{eq_LMIstable}, 
   respectively.
   
   % 
   % \marginJC{This is worded a bit confusingly: ``for all
   % $\overline{R}(\bm{\epsilon},\bm{r})$'' does not striclty speaking
   % make sense, b/c we have defined
   % $\overline{R}(\bm{\epsilon},\bm{r})$, and it's only 1 matrix, not
   % many matrices. I think what you want to explain (see earlier
   % margin where I ask you to expand on relationship between this
   % matrix and the $R_i$'s) is that the structure of this matrix
   % function will change, as a function of the state, and therefore
   % one can interpret it as a switched system. So here we do an
   % overapproximation, and instead consider arbitrary matrices
   % $R$. You can probably write it better than what I say here, but
   % this gives the basic idea, no? }
   % 
   % \marginXW{Addressed!}
   % 
   % \marginJC{As written, (29) is exactly the same as (28), and it does
   % not convey what you really want to say.}
   % 
   % Further recall that \eqref{eq_dstmodels0}
   To continue, recall that matrices $M,N$ need to satisfy
   constraints~\eqref{eq_LQMLQN}-\eqref{eq_LQMLQN2} and
   \eqref{eq_ZMNC}. Since $\overline{P}\succ 0$, $S_1= M\overline{P}$
   and $S_2=N\overline{P}$, these constraints can be equivalently
   written as
   \begin{align*}
     \mathcal{L}Q^{\top}M\overline{P}
     &=\mathcal{L}Q^{\top}S_1 =
       \bm{0}_{n(n+m)\times n},
     \\
     \mathcal{L}Q^{\top}N \overline{P}
     & = \mathcal{L}Q^{\top}S_2
      = \bm{0}_{n(n+m)\times n}  ,
    \\
    C_1Q^{\top}M\overline{P}
    &=C_1Q^{\top}S_1=\overline{P},
    \\
    C_1Q^{\top}N\overline{P} &=C_1Q^{\top}S_2=\bm{0}_{n \times n},
    \\
    Z(M+N)\overline{P} &=Z(S_1+S_2)=(1-\alpha)\overline{P} ,
  \end{align*}
  which correspond to~\eqref{eq_LCstable1}-\eqref{eq_LCstable3}.
  % 
  % \marginJC{This needs some adjusting: the connection between
  % $\overline{R}(\bm{\epsilon},\bm{r})$ and the $R_i$'s should be
  % much earlier (see margin before the th statement), (29) should be
  % written for any $R$,...}  \marginXW{Addressed.}
  % 
\end{proof}

Note that Theorem~\ref{Thm_stable} only provides a sufficient
condition for stabilizing the system~\eqref{eq_dstmodel}
to~$\bm{r}$. The conservativeness stems from the facts that we search
for a \emph{quadratic} Lyapunov function and that, in the derivation
from~\eqref{eq_Lypieq} to~\eqref{eq_lmir}, we ignore the dependency of
$\overline{R}(\bm{\epsilon},\bm{r})$ on~$\bm{\epsilon}$. This means
that, even if \eqref{eq_lmir} does not hold, \eqref{eq_Lypieq} may
still be true.

\begin{remark}\longthmtitle{Computational complexity of solving LMIs}
  LMIs with linear constraints can be efficiently solved using
  existing algorithms, cf.~\cite{MG-SB:14-cvx}.  The computational
  complexity of solving an equation of the form~\eqref{eq_LMIstable}
  is polynomial in $T_d$, which is the number of data samples and
  determines the sizes of matrices $Z$ and $Q$.  However, due to the
  combinatorial nature in the definition of the vertices $R_j$, the
  number of equations in~\eqref{eq_LMIstable} that need to be solved
  is of order $2^n$,
  % %
  % \marginJC{of order $2^n$, no? I think it's actually $2^n-1$, no?}
  % %
  % \marginXW{Addressed}
  which grows exponentially with the dimension of the system. We
  address this issue later in Section~\ref{Sec_ReducedC} by providing
  a sufficient condition with reduced computational
  complexity. \oprocend
\end{remark}

It is known~\cite{RCD-RHB:11} that feed-forward loops are non-robust
to system disturbances, meaning that slight changes in system
parameters may lead to a large degradation in tracking performance.
To address this issue, in the next section, we take advantage of the
classical idea of integral feedback and design an augmented feedback
controller with error integration for better robustness against
disturbance.

\section{Data-Driven Control with Error
  Integration}\label{Sec_INTControl}%
In this section, we design an augmented feedback controller with error
integration in the following form
\begin{subequations}\label{eq_intctrl}
  \begin{align}\label{eq_int_u}
    &\bm{u}(t)= K_1(\bm{x}(t)-\bm{r})+ K_2\bm{\xi}(t)
    \\
    \label{eq_int_xi}
    &\bm{\xi}(t+1)= \bm{\xi}(t)+ (\bm{x}(t)-\bm{r})
  \end{align}
\end{subequations} to solve Problem \ref{Prob}, where
$K_1\in\mathbb{R}^{m\times n}$ and $K_2\in\mathbb{R}^{m\times n}$ are
controller gains, and $\bm{\xi}(t)\in\mathbb{R}^n$ is an integrator
that accumulates the system error.  Compared with the linear feedback
controller \eqref{eq_theproblem}, the integral controller offers
better robustness against state disturbances.
% However, alongside this advantage, we also note a limitation of this
% controller as follows.

\begin{remark}\longthmtitle{Limitations of the controller with error
    integration}
  The augmented feedback controller with error integration in the form
  of \eqref{eq_intctrl} is only applicable to a reference value which
  excludes the lower and upper bounds of system states,
  $0< \bm{r}< \frac{s}{1-\alpha}$.
  % at $0$ and $\frac{s}{1-\alpha}$.
  This limitation, which happens to all integral
  controls when applied to dynamics with threshold saturation, is
  caused by the following reason. The equilibrium of dynamics
  \eqref{eq_dstmodel} under the controller \eqref{eq_intctrl} requires
  the convergence of $\bm{\xi}(t)$ to a specific value
  $\bm{\xi}^{\star}$, which depends on the given $\bm{r}$ and the
  choice of $K_2$. If any entry of $\bm{r}$ equals $0$ or
  $\frac{s}{1-\alpha}$, then, due to the constraints
  $0 \leq \bm{x}(t) \leq \frac{s}{1-\alpha}$, the corresponding entry
  in $(\bm{x}(t) - \bm{r})$ is either always non-negative or always
  non-positive. Consequently, the monotonic dynamics of $\bm{\xi}(t)$
  cannot guarantee convergence to the desired $\bm{\xi}^{\star}$. This
  issue does not occur when $0 < \bm{r} < \frac{s}{1-\alpha}$. It is
  worth noting that, in engineering practice, it is rare to require
  the system states to reach their exact saturation bounds. \oprocend
\end{remark}

% {\color{purple} }

\subsection{Closed-loop data-based
  representation}\label{Sec_CLDR_INT}
Based on the results in Section~\ref{SEC_KI}, we first derive a
closed-loop data-based representation for system \eqref{eq_dstmodel}
under the controller~\eqref{eq_intctrl}.

%\begin{lemma}\longthmtitle{Closed-loop data-based
%		representation for integral feedback}\label{LM_CLSYS1_INT} 
%	Consider the feedback controller $\bm{u}(t)=K_1\bm{x}(t)+ K_2\bm{r}$. Let
%	Assumption \ref{Ass_rankQ} hold.  Then the closed-loop form of
%	system \eqref{eq_dstmodel} with controller \eqref{eq_intctrl} has
%	the data-based representation,	
%\end{lemma}

%Similar to the result in Sec. \ref{Sec_CLDR}, $\Gamma(\bm{x}(t),\bm{r},\bm{\xi})$ can be explicitly constructed by the following method.

\begin{lemma}\longthmtitle{Closed-loop data-based representation for
    augmented feedback controller with error
    integration}\label{LM_CLSYS_INT}
  Consider a controller in the form of \eqref{eq_intctrl}.  Given data
  matrices $Z$ and $Q$ in~\eqref{eq:new-data-matrices}, let
  Assumption~\ref{Ass_rankQ} hold. Then, the
  system~\eqref{eq_dstmodel} has the data-based representation,
  \begin{align}\label{eq_DDMZMtr_INT}
    \bm{x}(t+1) = \alpha \bm{x}(t) + \left[Z(M \bm{x}(t)+N \bm{r}+U
    \bm{\xi}(t))\right]_0^{s} ,
  \end{align}
  where $M, N, U\in\mathbb{R}^{nT_d\times n}$ satisfy
  \begin{subequations}\label{eq_QMK_INT}
    \begin{align}\label{eq_QMK_INTa}
      Q^{\top}M=
      &~\bm{1}_n\otimes
        \begin{bmatrix}
          I_n\\K_1
        \end{bmatrix},
        \quad Q^{\top}U=\bm{1}_n\otimes
        \begin{bmatrix}
          \bm{0}_{n\times n}
          \\
          K_2
        \end{bmatrix} ,
      \\
      \label{eq_QMK_INTb}
      & Q^{\top}N=\bm{1}_n\otimes
        \begin{bmatrix}
          \bm{0}_{n\times n}
          \\
          -K_1
        \end{bmatrix} .
    \end{align}
  \end{subequations}
\end{lemma}
\smallskip
\begin{proof}
  Based on Lemma~\ref{LM_OLSYS}, define a new function
  $\Gamma: \mathbb{R}^{n}\times \mathbb{R}^{n}\times
  \mathbb{R}^{n}\to\mathbb{R}^{n\times nT_d}$ satisfying:
  \begin{align*}
    \Gamma(\bm{x},\bm{r},\bm{\xi}) =
    F\left(
    \begin{bmatrix} 
      I_n\\K_1
    \end{bmatrix}
    \bm{x}
    +
    \begin{bmatrix} 
      \bm{0}_n\\-K_1
    \end{bmatrix}\bm{r}
    +\begin{bmatrix} 
       \bm{0}_n\\K_2
     \end{bmatrix}\bm{\xi}\right).
  \end{align*}
  Then equation~\eqref{eq_consts} yields:
  \begin{align}\label{eq_conststr_INT}
    \begin{aligned}
      &\Gamma(\bm{x},\bm{r},\bm{\xi})\cdot(I-E_d)\mathcal{P}_d
      \\
      =&
         I_{n}\!\otimes\!\left(\bm{x}^{\top}
         \begin{bmatrix}
           I_n & \!\!\!K_1^{\top}
         \end{bmatrix}
         \!+\!\bm{r}^{\top}
         \begin{bmatrix}
           \bm{0}_{n\times n} & \!\!\!-K_1^{\top}
         \end{bmatrix}
         \!+\!\bm{\xi}^{\top}
         \begin{bmatrix}
           \bm{0}_{n\times n} & \!\!\!\!K_2^{\top}
         \end{bmatrix}\right),
    \end{aligned}	
  \end{align}
  and the closed-loop data-based representation~\eqref{eq_DDMG} can be
  written as
  \begin{align}\label{eq_DDMGtr_INT}
    \bm{x}(t+1)=\alpha
    \bm{x}(t) +
    \left[\Gamma(\bm{x}(t),\bm{r},\bm{\xi})\cdot(I-E_d)
    \mathcal{Z}_d\right]_0^{s}.     
  \end{align}
  Based on \eqref{eq_conststr_INT}-\eqref{eq_DDMGtr_INT}, and
  similarly to the proof of Lemma~\ref{LM_CLSYS}, one can associate
  matrices $M$, $N$, $U$ with terms $\bm{x},\bm{r},\bm{\xi}$,
  respectively, and obtain equations
  \eqref{eq_DDMZMtr_INT}-\eqref{eq_QMK_INT} as counterparts of
  equations \eqref{eq_DDMZMtr}-\eqref{eq_QMK}. We omit the details for
  brevity.
  %{\color{red} The proof follows from the results of Lemmas \ref{LM_CLSYS1} and
  %  \ref{LM_CLSYS}},
  %
  % \marginJC{This wording is confusing: i don't think it can follow
  %   from those lemmas, as those results are valid for the other type
  %   of controller. Maybe you mean that the argument is similar? But I
  %   think it really follows from Lemma 3.1, no?}
  %   \marginXW{Rewritten this part and provided a bit more details.}
  %
  %by defining a new function
  %$\Gamma: \mathbb{R}^{n}\times \mathbb{R}^{n}\times
  %\mathbb{R}^{n}\to\mathbb{R}^{n\times nT_d}$ for
  %\begin{align}\label{eq_DDMGtr_INT}
  %  \bm{x}(t+1)=\alpha
  %  \bm{x}(t) +
  %  \left[\Gamma(\bm{x}(t),\bm{r},\bm{\xi})\cdot(I-E_d)
  %  \mathcal{Z}_d\right]_0^{s}     
  %\end{align}
  %where
  %\begin{align}\label{eq_conststr_INT}
  %  \begin{aligned}
  %    &\Gamma(\bm{x},\bm{r},\bm{\xi})\cdot(I-E_d)\mathcal{P}_d
  %    \\
  %    =&
  %        I_{n}\!\otimes\!\left(\bm{x}^{\top}
  %        \begin{bmatrix}
  %          I_n & \!\!\!K_1^{\top}
  %        \end{bmatrix}
  %        \!+\!\bm{r}^{\top}
  %        \begin{bmatrix}
  %          \bm{0}_{n\times n} & \!\!\!-K_1^{\top}
  %        \end{bmatrix}
  %        \!+\!\bm{\xi}^{\top}
  %        \begin{bmatrix}
  %          \bm{0}_{n\times n} & \!\!\!\!K_2^{\top}
  %        \end{bmatrix}\right).
  %  \end{aligned}	
  %\end{align}
  %One then associates matrices $M,N,U$ with terms
  %$\bm{x},\bm{r},\bm{\xi}$, respectively, and obtain equations
  %\eqref{eq_DDMZMtr_INT}-\eqref{eq_QMK_INT} as counterparts of
  %equations \eqref{eq_DDMZMtr}-\eqref{eq_QMK}. We omit the details for
  %brevity.
\end{proof}

In equation~\eqref{eq_QMK_INT}, the coupled constraints introduced by
the Kronecker product can be equivalently reformulated as follows.

%{\color{purple}Comparing the statements in Lemmas~\ref{LM_CLSYS1} and~\ref{LM_CLSYS}, we see that finding $\Gamma: \mathbb{R}^{n}\times \mathbb{R}^{n}\to\mathbb{R}^{n\times nT_d}$ satisfying~\eqref{eq_conststr} can be accomplished by finding the matrix $M, N\in\mathbb{R}^{nT_d\times n}$ satisfying~\eqref{eq_QMK}. Regarding the solution of this latter equation, we note that when $K_1, K_2$ are given, $M, N$ can be readily computed. In fact, equations~\eqref{eq_QMIK} prescribe exactly how to find $M_i, N_i$ for each $i\in\until{n}$.  However, since our ultimate goal is to design controller gain matrices $K_1, K_2$ themselves, the solution in $M, N$ and $K_1, K_2$ of equation~\eqref{eq_QMK} poses the challenge of jointly solving for $n$ coupled systems of linear equations.  The following result provides a reformulation of the equation showing that $K_1, K_2$ can be expressed as a function of~$M$.}

\begin{lemma}\longthmtitle{Decoupling constraints for closed-loop
    data-based representation}\label{LM_Decouple_INT}
  Given $L\in\mathbb{R}^{n\times n}$,
  $C_1 \in \mathbb{R}^{n\times n(n+m)}$ and
  $C_2\in\mathbb{R}^{m\times n(n+m)}$ as in Lemma~\ref{LM_decouple},
  the constraint~\eqref{eq_QMK_INTa} can be equivalently written as
  \begin{subequations}\label{eq_matrix_constraints_INT}
    \begin{align}
      \mathcal{L}Q^{\top}M
      &=\bm{0}_{n(n+m)\times n}
        ,\quad\mathcal{L}Q^{\top}U=\bm{0}_{n(n+m)\times
        n}  \label{eq_LQMLQN_INT} 
      \\
      C_1Q^{\top}M
      &=I_n,\quad\quad\quad\quad~~ C_1Q^{\top}U=\bm{0}_{n\times n}
        \label{eq_LQMLQN2_INT}
      \\
      C_2 Q^{\top}M
      &=K_1,\quad\quad\quad\quad~ C_2 Q^{\top}U=K_2\label{eq_KQM_INT}
    \end{align}
  \end{subequations}
\end{lemma}
\begin{proof}
  Given the shared structure of \eqref{eq_QMK_INTa} and
  \eqref{eq_QMK}, the derivation of the equations
  \eqref{eq_matrix_constraints_INT} follows directly from
  Lemma~\ref{LM_decouple}.
  % Now, from the diagonal structure of $Q$, \eqref{eq_QMK_INTa} can be equivalently rewritten as
  %	\begin{align}\label{eq_QMIK_INT}
  %   Q_i^{\top}M_i=\begin{bmatrix}
  %     I_n\\K_1
  %			\end{bmatrix} ,~~
%		Q_i^{\top}U_i=\begin{bmatrix}
%			\bm{0}_n\\K_2
%		\end{bmatrix}
%	\end{align}
%	for $i\in\until{n}$, where $M=\col\{M_1,\cdots,M_n\}$, $U=\col\{U_1,\cdots,U_n\}$ with
%	$M_i, U_i\in\mathbb{R}^{T_d\times n}$. 
%From the definition of $L$, one has $\ker (L)= \image
%(\bm{1}_n)$. Thus, $\ker(\mathcal{L})=\image (\bm{1}_n\otimes
%I_{n+m})$. From \eqref{eq_LQMLQN_INT}, for $i,j\in\until{n}$, there
%holds
%\begin{align}\label{eq_QMQM_INT}
%	Q_i^{\top}M_i = Q_j^{\top}M_j,~~ Q_i^{\top}U_i = Q_j^{\top}U_j .
%\end{align}
%Furthermore, from (\ref{eq_matrix_constraints_INT}b-c), one has 
%\begin{align}\label{eq_QMIK2_INT}
%	Q_1^{\top}M_1=\begin{bmatrix}
%		I_n\\K_1
%	\end{bmatrix} ,~~ Q_1^{\top}U_1=\begin{bmatrix}
%	\bm{0}_n\\K_2
%\end{bmatrix}.
%\end{align}
%Thus, equations \eqref{eq_QMQM_INT} and \eqref{eq_QMIK2_INT} are equivalent
%to \eqref{eq_QMIK_INT} for $i\in\until{n}$, which corresponds
%to~\eqref{eq_QMK_INTa}. Equation~(\ref{eq_matrix_constraints_INT}c) readily follows by left
%multiplying~\eqref{eq_QMK_INT} by $C_2$. 
\end{proof}

Note that Lemma~\ref{LM_Decouple_INT} only characterizes the
constraint \eqref{eq_QMK_INTa} and ignores \eqref{eq_QMK_INTb}. This
is because $Q^{\top}$ has full row rank.  Then, based on the
Rouch\'{e}-Capelli Theorem~\cite{RAH-CRJ:12}, for any $K_1$ in
\eqref{eq_QMK_INTb}, there always exists $N$ such that the equation
holds. Furthermore, since $K_1$ can be determined from $M$, the
constraints~\eqref{eq_QMK_INT} can be simplified to only considering
\eqref{eq_QMK_INTa}. The control matrices $K_1$ and $K_2$ for
\eqref{eq_intctrl} can therefore be designed by determining matrices
$M,U$ satisfying~\eqref{eq_matrix_constraints_INT}.

\subsection{LMI-based design of integral feedback gain matrices}
We follow an approach similar to that of
Section~\ref{Subsec_feedback_LMI} to propose a data-driven approach to
design the feedback gain matrices in~\eqref{eq_intctrl} for solving
Problem \ref{Prob}.  We start by letting
$\bm{\epsilon}(t)=\bm{x}(t)-\bm{r}$. Based on \eqref{eq_intctrl} and
\eqref{eq_DDMZMtr_INT}, one has
\begin{align*}
  \bm{\epsilon}(t+1)
  &=\alpha \bm{\epsilon}(t) + \left[Z(M \bm{x}(t)+N
    \bm{r}+U \bm{\xi}(t))\right]_0^{s}
    -(1-\alpha)\bm{r} ,
  \\ 
  \bm{\xi}(t+1)&=\bm{\xi}(t)+\bm{\epsilon}(t).
\end{align*}
For now, we make the assumption that $ZU$ is non-singular (we show
later in Proposition \ref{Prop_nonsingular} that this property
actually holds). Based on this assumption, for any $\bm{r}$, there
must exist $\bm{\xi}^{\star}$ such that
\begin{align}\label{eq_existxi} 
  Z((M+N)\bm{r}+U\bm{\xi}^{\star})&=(1-\alpha)\bm{r} .
\end{align}
Using this $\bm{\xi}^{\star}$, let
$\bm{e}(t)=\bm{\xi}(t)-\bm{\xi}^{\star}$. Then we have
\begin{align}\label{eq_dstmodels_INT} 
  \bm{\epsilon}(t+1)
  % &=\alpha \bm{\epsilon}(t) + [Z(M \bm{x}(t)+N \bm{r}+U \bm{\xi}(t))+(1-\alpha)\\
  % &~~~~-Z((M+N)\bm{r}+U\bm{\xi}^{\star})\bm{r}]_0^{s} -(1-\alpha)\bm{r}\\
	&=\alpha \bm{\epsilon}(t) + \left[ZM \bm{\epsilon}(t)+ZU
          \bm{e}(t)+(1-\alpha)\bm{r}\right]_0^{s} \nonumber
  \\
    &~~~~ - (1-\alpha)\bm{r}  \nonumber
  \\
    &=\alpha \bm{\epsilon}(t) +
      \overline{R}(\bm{\epsilon},\bm{e},\bm{r})(ZM \bm{\epsilon}(t)+ZU
      \bm{e}(t)) , \nonumber
  \\
  \bm{e}(t+1)&=\bm{e}(t)+\bm{\epsilon}(t) ,
\end{align}
where
$\overline{R}(\bm{\epsilon},\bm{e},\bm{r})\in\mathbb{R}^{n\times n}$
is a diagonal matrix with each entry
$\overline{R}(\bm{\epsilon},\bm{e},\bm{r})[i,i]$ defined as
\begin{align*}
  \begin{cases}
    \text{\small$\frac{(\left[ZM \bm{\epsilon}+ZU
    \bm{e}+(1-\alpha)\bm{r}\right]_0^{s} - (1-\alpha)\bm{r})[i]}{(ZM
    \bm{\epsilon}+ZU \bm{e})[i]}\mkern-80mu$}
    &
    \\
    &
      \text{if {\small$(ZM \bm{\epsilon}+ZU
      \bm{e}+(1-\alpha)\bm{r})[i]>s$}}
    \\
    &\text{or
      {\small$(ZM
      \bm{\epsilon}+ZU
      \bm{e}+(1-\alpha)\bm{r})[i]<0$}}
    \\ 
    1 &\text{otherwise}
  \end{cases}
\end{align*}
Note that $\overline{R}(\bm{\epsilon},\bm{e},\bm{r})[i,i]$ is always
well defined for any $\bm{\epsilon},\bm{e}$, and
$0< \bm{r}< \frac{s}{1-\alpha}$, because the first two conditions hold
only if $(ZM \bm{\epsilon}+ZU \bm{e})[i]\neq 0$.  It is also worth
noting that, in \eqref{eq_dstmodels_INT}, the matrix $N$ does not
appear. This aligns with the fact that the
equations~\eqref{eq_matrix_constraints_INT} do not involve $N$, and
allow us to design $K_1, K_2$ by only considering $M, U$.

To continue, similar to the analysis in
Section~\ref{Subsec_feedback_LMI}, we have
$0< \overline{R}(\bm{\epsilon},\bm{e},\bm{r})[i,i]\le 1$, for all
$i\in\until{n}$.  We employ the same matrices
$R_j \in \real^{n \times n}$, $j\in\until{2^n}$, to create a convex
combination:
\begin{align}\label{eq_Rhull_INT}
    \overline{R}(\bm{\epsilon},\bm{e},\bm{r})= \sum_{j=1}^{2^n}\mu_j
  R_j,~~\text{with}~~\sum_{j=1}^{2^n}\mu_j =1 
    ~~\text{and}~~\mu_1>0.
\end{align}
The following result provides a way to synthesize the controller
leveraging \eqref{eq_Rhull_INT}.
% Define
% $$\mathcal{R}=\text{Cov}(R_i) \quad\text{for}\quad
% i\in\until{2^n}.$$

\begin{theorem}\longthmtitle{Data-driven synthesis for augmented
    feedback controller with error integration via
    LMIs}\label{Thm_stable_INT}
  Given $\overline P \in \mathbb{R}^{2n\times2n}$, consider the block
  decomposition
  \begin{align*}
    \overline P=
    \begin{bmatrix}
      \overline P_{11} & \overline P_{12}
      \\
      \overline P_{12}^{\top} & \overline P_{22}
    \end{bmatrix}
  \end{align*}
  with
  $\overline{P}_{11},\overline{P}_{12},\overline{P}_{22}\in\mathbb{R}^{n\times
    n}$. Let the matrices $\overline P$ and
  $S_1,S_2\in\mathbb{R}^{nT_d\times n}$ satisfy
  \begin{subequations}\label{eq_LMIs_INT}
    \begin{align}\label{eq_lmistable_INT}
      \begin{bmatrix}
        \overline P & *
        \\
        \begin{pmatrix}
          \alpha \overline P_{11} \! +\!  ZS_1
          &  \alpha \overline
            P_{12} \! +\!  ZS_2
          \\
          \overline P_{11} \!+\! \overline P_{12}^{\top}
          & \overline
            P_{12}
            \!+\!
            \overline
            P_{22} 
        \end{pmatrix} & \overline P
      \end{bmatrix}&\succ 0,
      \\  
      \begin{bmatrix}\label{eq_lmistable2_INT}
        \overline P & *\\
        \!\begin{pmatrix}
            \alpha \overline P_{11} \!+\! R_jZS_1
            &  \!\!\alpha
              \overline P_{12}
              \!+\! R_jZS_2
            \\
            \overline P_{11} \!+\! \overline P_{12}^{\top}
            & \overline
              P_{12}
              \!+\!
              \overline
              P_{22} 
          \end{pmatrix} &\!\!\! \overline P
      \end{bmatrix}&\succeq 0,
      \\
      \mathcal{L}Q^{\top}S_1=\bm{0}_{n(n+m)\times
      n},~~~~\mathcal{L}Q^{\top}S_2=\bm{0}_{n(n+m)\times
      n}& \label{eq_LCstable1_INT} 
      \\
      C_1Q^{\top}S_1=\overline P_{11},~~~~~C_1Q^{\top}S_2=\overline
      P_{12},~~~& \label{eq_LCstable2_INT} 
    \end{align}
  \end{subequations}
  for $R_j$, $j\in\{2,3,\dots,2^n\}$, where $*$ represents the
  symmetric part of the matrix.  Then the controller~\eqref{eq_intctrl}, with
  \begin{align}\label{eq_KQS_INT}
    \begin{bmatrix}
      K_1&K_2
    \end{bmatrix}=C_2 Q^{\top}
    \begin{bmatrix}
      S_1&S_2
    \end{bmatrix}\overline P^{-1},
  \end{align}
  ensures that $(\bm{r},\bm{\xi}^{\star})$ is asymptotically 
  stable for the closed-loop system.
\end{theorem}
% %
% \marginJC{Shouldn't we say `` ensures that $(\bm{r},\bm{\xi}^{\star})$
%   is asymptotically stable for the closed-loop system.'' ?}
%   \marginXW{Addressed.}
% %
\begin{proof}
  To prove the statement, we construct a quadratic function of the
  form
  \begin{align}
    V(\bm{\epsilon},\bm{e})
    =
    \begin{bmatrix}
      \bm{\epsilon}\\
      \bm{e}
    \end{bmatrix}^{\top} P
    \begin{bmatrix}
      \bm{\epsilon}\\
      \bm{e}
    \end{bmatrix}
  \end{align}
  with $P\succ0$. The convergence result holds if, for any
  $\col\{\bm{\epsilon}(t), \bm{e}(t)\}\neq 0$, the function satisfies
  \begin{align*}
    V(\bm{\epsilon}(t+1),\bm{e}(t+1))<V(\bm{\epsilon}(t),\bm{e}(t)).
  \end{align*}
  From \eqref{eq_dstmodels_INT}, this requires 
  \begin{align}%\nonumber	
    &\begin{bmatrix}
       \bm{\epsilon}\\
       \bm{e}
     \end{bmatrix}^{\top}
      \Phi(\bm{\epsilon},\bm{e},\bm{r})^{\top}
      P\Phi(\bm{\epsilon},\bm{e},\bm{r})
      \begin{bmatrix}
        \bm{\epsilon}\\
        \bm{e}
      \end{bmatrix}
      <\begin{bmatrix}
         \bm{\epsilon}\\
         \bm{e}
       \end{bmatrix}^{\top} P\begin{bmatrix}
                               \bm{\epsilon}\\
                               \bm{e}
                             \end{bmatrix}
      \label{eq_Lypieq_INT}
  \end{align}
  when either $\bm{\epsilon}\neq 0$ or $\bm{e}\neq 0$, where
  \begin{align}\label{eq_defphi}%\nonumber	
    \Phi(\bm{\epsilon},\bm{e},\bm{r})
    =
    \begin{bmatrix}
      \alpha I + \overline{R}(\bm{\epsilon},\bm{e},\bm{r})ZM
      &
        \overline{R}(\bm{\epsilon},\bm{e},\bm{r})ZU 
      \\
      I & I
    \end{bmatrix}.
  \end{align}
  To make sure \eqref{eq_Lypieq_INT} holds for the switched system, it
  is sufficient to guarantee the following matrix inequality holds
  \begin{align}\label{eq_lmi0_INT}
    \Phi(\bm{\epsilon},\bm{e},\bm{r})^{\top}
    P\Phi(\bm{\epsilon},\bm{e},\bm{r}) \prec  P
  \end{align}
  for all possible values of
  $\overline{R}(\bm{\epsilon},\bm{e},\bm{r})$.  Using the Schur
  complement~\cite{CKL-RM:04}, \eqref{eq_lmi0_INT} is equivalent to
  \begin{align}\label{eq_lmi1_INT}
    \begin{bmatrix}
      {P} &*\\
      \begin{pmatrix}
        \alpha I + \overline{R}(\bm{\epsilon},\bm{e},\bm{r})ZM
        &
          \overline{R}(\bm{\epsilon},\bm{e},\bm{r})ZU
        \\
        I & I
      \end{pmatrix} & {P}^{-1}
    \end{bmatrix}\succ 0.
  \end{align}
  Leveraging \eqref{eq_Rhull_INT}, it is sufficient to 
  consider instead
  \begin{subequations}\label{eq_lmir_INT}
    \begin{align}\label{eq_lmira_INT}
      \begin{bmatrix}
        {P} &*\\
        \begin{pmatrix}
          \alpha I + ZM &
                          ZU
          \\
      I & I
    \end{pmatrix} & {P}^{-1}
      \end{bmatrix}\succ 0,
      \\
    \begin{bmatrix}
       {P} &*\\
       \begin{pmatrix}
      \alpha I + R_jZM &
       R_jZU
      \\
      I & I
    \end{pmatrix} & {P}^{-1}
     \end{bmatrix}\succeq 0.
  \end{align}
  \end{subequations}
  for a common $P\succ0$, for all $j\in\{2,3,\dots,2^n\}$. Here
  \eqref{eq_lmira_INT} is obtained by letting $R_1=I_n$. Since in
  \eqref{eq_Rhull_INT}, $\mu_1>0$, the convex combination of
  \eqref{eq_lmir_INT} guarantees \eqref{eq_lmi1_INT}.

  Now, let
  $$\overline P=
  \begin{bmatrix}
    \overline P_{11} & \overline P_{12}\\
    \overline P_{12}^{\top} & \overline P_{22}
  \end{bmatrix}={P}^{-1}.
  $$ 
  Since $\overline{P}\succ 0$, without losing generality, we introduce
  a pair of matrices 
  $\begin{bmatrix} S_1 & S_2
   \end{bmatrix} \triangleq
   \begin{bmatrix}
     M& U
   \end{bmatrix}\overline{P}
   =
   \begin{bmatrix}
     M\overline P_{11}+U\overline P_{12}^{\top}
     & M\overline
       P_{12}+U\overline
       P_{22} 
   \end{bmatrix}$. 
   Pre- and post-multiplying~\eqref{eq_lmir_INT} by
   $\begin{bmatrix}
      \overline{P} &0\\
      0& I
    \end{bmatrix}$  yields 
    \eqref{eq_lmistable_INT}-\eqref{eq_lmistable2_INT}, respectively.
   
    To continue, recall that matrices $M,U$ need to satisfy
    constraints~\eqref{eq_LQMLQN_INT} and
    \eqref{eq_LQMLQN2_INT}. Since $\overline{P}\succ 0$, and
    $\begin{bmatrix} S_1 & S_2
     \end{bmatrix} =
   \begin{bmatrix}
     M& U
   \end{bmatrix}\overline{P}$, 
  these constraints can be equivalently 
  written as
   \begin{align*}
     \mathcal{L}Q^{\top}
     \begin{bmatrix}
       M&U
     \end{bmatrix}
     \overline{P}
     &=\mathcal{L}Q^{\top}
       \begin{bmatrix}
         S_1&S_2
       \end{bmatrix}
       =
         \bm{0}_{n(n+m)\times 2n}
     \\
     C_1Q^{\top}
     \begin{bmatrix}
       M&U
     \end{bmatrix}
     \overline{P}
     &=C_1Q^{\top}
       \begin{bmatrix}
         S_1&S_2
       \end{bmatrix}
       =
       \begin{bmatrix}
         \overline P_{11}&\overline P_{12}
       \end{bmatrix}
   \end{align*}
   which correspond to~(\ref{eq_LMIs_INT}c-d).
\end{proof}

%Theorem~\ref{Thm_stable} is only a sufficient
%condition for stabilizing system~\eqref{eq_dstmodel} to $\bm{r}$ with controller \eqref{eq_intctrl}. The
%conservativeness stems from the facts that we search for a
%\emph{quadratic} Lyapunov function and that, in the derivation
%from~\eqref{eq_Lypieq_INT} to~\eqref{eq_lmi0_INT}, we ignore the coupling
%between $\overline{R}(\bm{\epsilon},\bm{e},\bm{r})$ and $\bm{\epsilon},\bm{e}$. In fact, even if
%\eqref{eq_lmi0} does not hold, \eqref{eq_Lypieq} may still be true. Finally, we provide the following proposition to address an assumption we made on matrix $ZU$ when driving equation \eqref{eq_existxi}.

Theorem~\ref{Thm_stable_INT} relies on the fact that $ZU$ is a
non-singular matrix, so that $\bm{\xi}^{\star}$ is well defined. We
show that this holds next.

\begin{proposition}\longthmtitle{Non-singular
    matrix}\label{Prop_nonsingular}
  The matrix $ZU$, with $U$ obtained from
  Theorem~\ref{Thm_stable_INT}, is non-singular.
\end{proposition}
\begin{proof}
  We reason by contradiction. Assume $ZU$ is singular.  Note
  that~\eqref{eq_lmira_INT} corresponds to~\eqref{eq_lmi0_INT} when
  $\overline{R}(\bm{\epsilon},\bm{e},\bm{r})=I_n$.  Considering the
  block structure of $\Phi(\bm{\epsilon},\bm{e},\bm{r})$ in
  \eqref{eq_defphi}, since $ZU$ is singular, then at least one
  eigenvalue of $\Phi(\bm{\epsilon},\bm{e},\bm{r})$ is 1. This
  contradicts~\eqref{eq_lmi0_INT}, which requires all all eigenvalues
  of $\Phi(\bm{\epsilon},\bm{e},\bm{r})$ to have magnitude strictly less
  than~1.
  % Considering the blocked structure of
  % $\Phi(\bm{\epsilon},\bm{e},\bm{r})$ in \eqref{eq_defphi} with
  % $\overline{R}(\bm{\epsilon},\bm{e},\bm{r})=I_n$, which corresponds
  % to~\eqref{eq_lmira_INT}, if $ZU$ were singular, then at least one
  % eigenvalue of $\Phi(\bm{\epsilon},\bm{e},\bm{r})$ would be 1. This
  % contradiction implies that $ZU$ must be non-singular.
\end{proof}

\begin{remark}\longthmtitle{Computational complexity of solving LMIs,
    cont'd}
  The LMIs in \eqref{eq_LMIs_INT} have twice the dimensions compared
  with those in \eqref{eq_LMIs}. Apart from this, similar to the
  complexity of solving \eqref{eq_LMIstable}, the computational
  complexity of solving a single \eqref{eq_lmistable2_INT} is polynomial in
  $T_d$. However, the number of equations in~\eqref{eq_lmistable2_INT} that
  we need to solve is of order $2^n$,
  %         %
  % \marginJC{Same comment as above}
  % \marginXW{Addressed.}
  %         %
  which grows exponentially with the dimension of the system.  This
  observation motivates Section~\ref{Sec_ReducedC}, which provides a
  sufficient condition to reduce the computational complexity for both
  Theorems~\ref{Thm_stable} and \ref{Thm_stable_INT}. \oprocend
\end{remark}

\section{Sufficient stability conditions with reduced computational
  complexity}\label{Sec_ReducedC}

%Theorems~\ref{Thm_stable} and \ref{Thm_stable_INT} require solving a
%number of LMIs that scale combinatorially with $n$. 
The high computational cost of Theorems~\ref{Thm_stable} 
and \ref{Thm_stable_INT} arises from the use of matrices $R_j$ 
in~\eqref{eq_Rhull} and~\eqref{eq_Rhull_INT} to construct convex
combinations representing $\overline{R}(\bm{\epsilon},\bm{r})$ and
$\overline{R}(\bm{\epsilon},\bm{e},\bm{r})$, respectively.
To reduce this complexity, we propose using a smaller
set of matrices whose convex combination can still represent
$\overline{R}(\bm{\epsilon},\bm{r})$ and
$\overline{R}(\bm{\epsilon},\bm{e},\bm{r})$.

Towards this end, define $\widetilde{R}_0=\bm{0}_{n\times n}$.
For $k\in\until{n}$, let $\widetilde{R}_k \in \real^{n\times n}$ be a
diagonal matrix such that
\begin{align*}
  \widetilde{R}_k[j,j]\triangleq 
  \begin{cases}
    1& \text{for } j=k ,
    \\
    0& \text{for } j\neq k .
  \end{cases}
\end{align*}
The next result provides an alternative set of LMIs whose number
scales linearly with~$n$, that can be used to replace 
conditions~\eqref{eq_LMIstable} and~\eqref{eq_lmistable2_INT},
respectively. 

\begin{proposition}\longthmtitle{A sufficient condition with reduced
    computational complexity for solving LMIs}\label{Prop_RCC}
  \begin{enumerate}
  \item Let $\overline{P}\in\mathbb{R}^{n\times n}$ and
    $S_1,S_2\in\mathbb{R}^{nT_d\times n}$ satisfy
    \begin{align}\label{eq_reducedcom}
      \begin{bmatrix}
        \overline{P} & *
        \\
        \alpha \overline{P} + {n}\widetilde{R}_k ZS_1 & \overline{P}
      \end{bmatrix}\succeq 0
    \end{align}
    for all $\widetilde{R}_k$, $k\in\{0,1,\dots,n\}$. Then, these
    matrices satisfy~\eqref{eq_LMIstable} for all
    $i\in\{2,3,\dots,2^n\}$;
    %
    % \marginJC{How about (27a)? Does this mean we need to solve (27a)
    %   and (50)?}
    % \marginXW{Yes, we do. I added an explanation after the 
    %   proposition.}
    % 
  \item Let 
    \begin{align*}
      \overline P=
      \begin{bmatrix}
        \overline P_{11} & \overline P_{12}
        \\
        \overline P_{12}^{\top} & \overline P_{22}
      \end{bmatrix} \in \mathbb{R}^{2n\times 2n}
    \end{align*}
    with $\overline{P}_{11}$, $\overline{P}_{12}$,
    $\overline{P}_{22}\in \mathbb{R}^{n\times n}$ and
    $S_1,S_2\in\mathbb{R}^{nT_d\times n}$ satisfy
    \begin{align}\label{eq_LMIsRC_INT}
      \begin{bmatrix}
        \overline P & *
        \\
        \!
        \begin{pmatrix}
          \alpha \overline P_{11} \!+\! n\widetilde{R}_kZS_1
          &
            \!\!\alpha
            \overline
            P_{12}
            \!+\!
            n\widetilde{R}_kZS_2
          \\
          \overline P_{11} \!+\! \overline P_{12}^{\top}
          & \overline P_{12} \!+\! \overline P_{22}
        \end{pmatrix}
        &\!\!\! \overline P
      \end{bmatrix}
      &\succeq 0
    \end{align}
    for all $\widetilde{R}_k$, $k\in\{0,1,\dots,n\}$.  Then these
    matrices satisfy~\eqref{eq_lmistable2_INT} for all
    $i\in\{2,3,\dots,2^n\}$.
    %
    % \marginJC{Same question as above? Does this mean we need to solve
    %   (42a) and (51)?}
    %   \marginXW{Yes.}
    % 
  \end{enumerate}  
\end{proposition}
\begin{proof}
  For brevity, we only provide the proof for case (ii). The proof for
  (i) is analogous.
    % The two conditions for Theorems~\ref{Thm_stable} and
  % \ref{Thm_stable_INT} follow the same principle. Here we only provide
  % the proof for the second case. %Since
  % %$\overline{R}(\bm{\epsilon},\bm{r})$ is diagonal and
  % %$0< \overline{R}(\bm{\epsilon},\bm{r})[i,i]\le 1$, 
  Given $j\in\until{2^n}$, let $a_{jk}={R}_j[k,k]$,
  $k\in\until{n}$.  It follows that
  \begin{align*}
    {R}_j = \sum_{k=1}^n a_{jk}\widetilde{R}_k\quad
    \text{and}\quad \sum_{k=1}^n a_{jk} \le n,
  \end{align*}
  %                                                                      
  % \marginJC{Why from $k=0$ in the sum? There is no zero row/column in $R_i$!}
  % \marginXW{Addressed. }
  % 
  where the inequality holds because
  $a_{jk}={R}_j[k,k]\in\{0,1\}$. For convenience, let
  $\sigma_j = \sum_{k=1}^n a_{jk}$.
  
  For $k\in\until{n}$, we multiply \eqref{eq_LMIsRC_INT} by
  $\frac{a_{jk}}{n}$ and sum the inequalities over all $k$ to yield
  \begin{align}\label{eq_LMIRCC}
    \begin{bmatrix}
      \frac{\sigma_j}{n}\overline{P}
      & *
      \\
      \!\begin{pmatrix}
          \alpha  \frac{\sigma_j}{n}\overline P_{11} \!+\! R_jZS_1
          &
            \!\!\alpha
            \frac{\sigma_j}{n}\overline
            P_{12}
            \!+\!
            R_jZS_2
          \\ 
          \frac{\sigma_j}{n}(\overline P_{11} \!+\! \overline
          P_{12}^{\top})
          &
            \frac{\sigma_j}{n}(\overline P_{12} \!+\! \overline P_{22})
        \end{pmatrix} &\!\!\! & \frac{\sigma_j}{n}\overline{P}
    \end{bmatrix}\succeq 0 .
  \end{align}
  Consider \eqref{eq_LMIsRC_INT} for
  $\widetilde{R}_{k=0}=\bm{0}_{n\times n}$.  Since $\sigma_j\le n$, we
  have
  \begin{align}\label{eq_aP}
    \left(1-\frac{\sigma_j}{n}\right)
    \begin{bmatrix}
      \overline{P} & *
      \\
      \!\begin{pmatrix}
          \alpha \overline P_{11}
          &  \!\!\alpha \overline P_{12}
          \\
          \overline P_{11} \!+\! \overline P_{12}^{\top}
          & \overline
            P_{12}
            \!+\!
            \overline
            P_{22} 
        \end{pmatrix} & \overline{P}
    \end{bmatrix}\succeq 0 .
  \end{align}
  Thus, for $j\in \{2,3,...,2^n\}$, adding \eqref{eq_LMIRCC} and
  \eqref{eq_aP} results in \eqref{eq_lmistable2_INT}.
  % 
  % \marginJC{At the beginning of the proof, we had $i\in\until{2^n}$,
  %   so we also get (42a), no?}
  %   \marginXW{when $i=1$, (now $j=1$), combining equations 
  %   \eqref{eq_LMIRCC} and \eqref{eq_aP} only results in a 
  %   semi-positive definite condition.
  % We still need to enforce the strictly positive definite condition 
  % for (42a).}
\end{proof}

The conditions in Proposition~\ref{Prop_RCC} are more computationally
tractable: we have only order $n$ LMIs to consider, instead of order
$2^n$ from~\eqref{eq_LMIstable} or~\eqref{eq_lmistable2_INT}.
%
% \marginJC{Again, check the numbers. I see $n+1$, not $n$. Also,
%   $2^n-1$, not $2^n$. Specific numbers are not that relevent, it's
%   really the order. So if you say ``order $n$'', then ok.}
% \marginXW{Addressed!}
% 
However, these conditions are stricter due to the rescaling factor on
the off-diagonal elements of the matrix.  As a consequence, there
might exist matrices satisfying~\eqref{eq_LMIstable}
or~\eqref{eq_lmistable2_INT}, but not the
corresponding~\eqref{eq_reducedcom} or \eqref{eq_LMIsRC_INT}.
Finally, note that when solving the LMIs,
conditions~\eqref{eq_LMIstable0} and \eqref{eq_lmistable_INT} still
need to be considered because they impose positive definite
conditions, whereas from the proof of Proposition~\ref{Prop_RCC} we
can only guarantee positive semi-definiteness for~$j=1$.

% %
% \marginJC{Adjust this addign (27a) and (42a) depending on how margins
%   above are resolved.}  \marginXW{Added.}
% %

%Clearly, the convex hull defined by the vertices $R_i$ for
%$i\in\until{2^n}$ is a subset of the convex hull defined by the
%vertices $\{\bm{0}_{n\times n}\}\bigcup \{\widetilde{R}_k, k=1,\cdots,n\}$. This
%allows us to reduce the $2^n$ number of LMIs in~\eqref{eq_LMIstable}
%and \eqref{eq_lmistable2_INT} to $n$ number of LMIs
%in~\eqref{eq_reducedcom} and \eqref{eq_LMIsRC_INT},
%respectively. However, because the new convex hull is larger than the
%original one, condition~\eqref{eq_reducedcom} is more strict. In other
%words, there might exist matrices satisfying~\eqref{eq_LMIstable}
%or~\eqref{eq_lmistable2_INT} but not the
%corresponding~\eqref{eq_reducedcom} or \eqref{eq_LMIsRC_INT}.

\section{Case Studies: Validating Data-Driven Control in Biological Systems}

In this section, we present two biological
examples~\cite{EN-JC:21-tacII,JF-JC-SS-PS:19} to validate the
effectiveness of the proposed data-driven control. The first example
tackles the regulation of neuro-excitation levels in rodents' brains
during selective listening tasks. This example assumes that the
excitation levels of specific neuronal populations can be directly
manipulated as system inputs The second example simulates the
regulation of human arousal levels using a brain–computer interface
(BCI) for a sensory-motor task.  This example builds on a more
realistic application with the potential for real-world
implementation.
% One experiment better than the other.
%
% \marginJC{I haven't read the subsections below when writing this
%   margin, but I'm curious as to why we have 2 different regulation
%   examples. I understand 2 is better than 1, but it would be nice if
%   we explain what one example adds with respect to the other.  }
% \marginXW{Added a comparison of the two experiments.}
% 

\subsection{Data-driven Regulation of Neuro-Excitation Levels in Rodents' Brain}
We consider an experiment that studies the regulation of neural
activation levels in rodents' brains for selective listening
tasks~\cite{EN-JC:21-tacII,CCR-MRD:14,CCR-MRD:14-crcns}.  The rodents
are exposed to a (left/right) white noise burst and a (high/low pitch)
narrow-band warble. Depending on the task, they need to focus on one
of them and ignore the other.  During the experiments, the firing
rates of the neuron cells are recorded from two different regions of
their brains: the prefrontal cortex (PFC) and the primary auditory
cortex (A1).
%
% \marginJC{This is too sketchy. I'd expand a bit to explain that they
%   are exposed to two stimuli: left-right warble sounds and high-low
%   pitch sounds and, depending on the task, they need to focus on one
%   of them and ignore the other.}  \marginXW{Expanded.}
% %

\begin{figure}[htb]
  \centering
  \includegraphics[width=\linewidth]{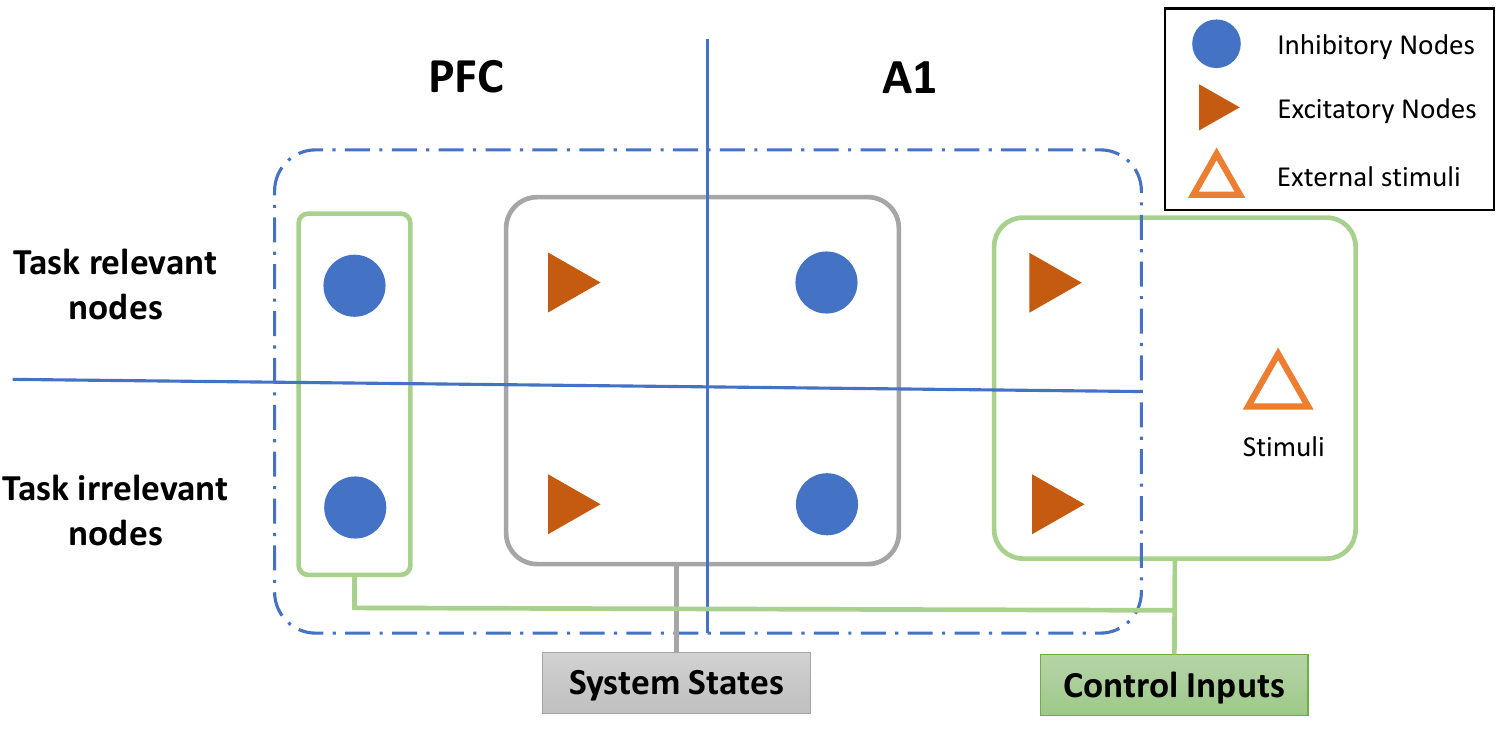}
  \caption{The model includes 8 groups of neuron cells and one
    external stimulus.  The nodes in the gray box are considered as
    system states; the nodes in green boxes are considered as control
    inputs.}\label{Fig_nn}
\end{figure}

We classify the neurons in the rodent brains into $8=2^3$ groups based
on a combination of: region (PFC, A1); type (excitatory, inhibitory);
and encoding (task relevant, irrelevant), as shown in
Fig.~\ref{Fig_nn}.  We choose 4 nodes as system states ($n=4$),
corresponding to the excitatory nodes in the PFC area and the
inhibitory nodes in the A1 area. The other 4 nodes along with the
external stimuli are considered as system inputs. Building on this
configuration, we have identified the following system parameters in
our previous work \cite{XW-JC:25-tac}:
\begin{align*}
  \alpha&=0.9728, \qquad s=0.3984, \\
  W&=\begin{bmatrix*}[r]
       0 &   0.0427 &  -0.0122  & 0\\
       0.0084 &   0 &  -0.0003 &  -0.0009\\
       0.0421 &   0.0334 &        0  & 0\\
       0.1031 &   0.0114 &  -0.0036  &       0
     \end{bmatrix*}
     ,
  \\
  B&=\begin{bmatrix*}[r]
       0.0114   &-0.0005 &  -0.0749  & -0.0017   & 0\\
       -0.0270  &  0.0015  &  0.2107  & 0  &  0\\
       -0.6332  &  0.0044  & -0.2840  & 0  &  0.0358\\
       -0.7236  &  0.0162  &  0.5482  & 0  &  0.0207
     \end{bmatrix*} .
\end{align*}
The upper bound of the system state is $\frac{s}{1-\alpha}=14.647$.
Note that the system matrices are assumed to be unknown during the
data-driven controller design -- we only use them to generate data
samples.  Based on $W$, $B$, $\alpha$ and $s$, we create data samples
with the following discrete-time system model
\begin{align*}
  \widetilde{\bm{x}}_d^+(k)=\alpha (\widetilde{\bm{x}}_d(k)) +
  \left[W(\widetilde{\bm{x}}_d(k))+B(\widetilde{\bm{u}}_d(k))\right]_0^{s}
\end{align*}
with $k\in\{1,\dots,T_d\}$ and $T_d=250$. For different $k$, system
states $\widetilde{\bm{x}}_d(k)$ and inputs $\widetilde{\bm{u}}_d(k)$
are chosen independently, i.e., the entries of
$\widetilde{\bm{x}}_d(k)$ are randomly chosen from $[0 ~~14.647]$; the
entries of $\widetilde{\bm{u}}_d(k)$ are randomly chosen from
$[0 ~~10]$, with uniform distributions. By validation, the data
samples satisfy the rank condition in Assumption \ref{Ass_rankQ}.

To validate Theorem~\ref{Thm_stable}, we consider the following
semi-definite programming (SDP) problem with a random reference input $r=
\begin{bmatrix}
  8.26& 4.42& 10.99& 6.95
\end{bmatrix}^{\top}$,
\begin{subequations}\label{eq_simuSDP}
  \begin{alignat}{2}
    & \text{maximize} \quad  && \gamma \label{eq_simuobj}
    \\
    & \text{subject to} \quad && \overline{P}\preceq I_n,
    \\
    &&& 
        \begin{bmatrix}
          \overline{P} & ~~*\\
          \alpha \overline{P} + ZS_1 & \overline{P}
        \end{bmatrix} \succeq  \gamma I_{2n} ,
        \label{eq_simulmi}
    \\
    &&& 
        \begin{bmatrix}
          \overline{P} & ~~*\\
          \alpha \overline{P} + {n}\widetilde{R}_k ZS_1 & \overline{P}
        \end{bmatrix} \succeq  0,
        \label{eq_simulmi2}
    \\  
    &&&  \mathcal{L}Q^{\top}S_1=\bm{0},\quad
        ~~~\mathcal{L}Q^{\top}S_2=\bm{0} ,
    \\
    &&&  C_1Q^{\top}S_1=\overline{P},\quad ~C_1Q^{\top}S_2=\bm{0} ,
    \\
    &&&  Z(S_1+S_2)=(1-\alpha)\overline{P},
  \end{alignat}
  % \begin{align}
  %   \text{maximize} \quad \gamma \quad &\label{eq_simuobj}
  %   \\		
  %   \overline{P}&\preceq I_n\\  	
  %   \begin{bmatrix}
  %     \overline{P} & ~~*\\
  %     \alpha \overline{P} + {n}\widetilde{R}_k ZS_1 & \overline{P}
  %   \end{bmatrix}&\succeq  \gamma I_n \label{eq_simulmi}  \\  
  %   \mathcal{L}Q^{\top}S_1=\bm{0},\quad ~~~\mathcal{L}Q^{\top}S_2&=\bm{0}\\
  %   C_1Q^{\top}S_1=\overline{P},\quad ~C_1Q^{\top}S_2&=\bm{0}\\
  %   Z(S_1+S_2)=(1-\alpha)\overline{P}
  % \end{align}
\end{subequations}
for all $\widetilde{R}_k$, $k\in\until{n}$. With the solution, we
define $K_1=C_2 Q^{\top}S_1 \overline{P}^{-1}$ and
$K_2=C_2 Q^{\top}S_2 \overline{P}^{-1}$ to design a controller
\begin{align*}
  \bm{u}(t)=K_1\bm{x}(t)+K_2\bm{r}.
\end{align*}
With $\gamma$ positive, conditions (\ref{eq_simuSDP}c-g) correspond to
Theorem~\ref{Thm_stable} and Proposition~\ref{Prop_RCC}(i) (for the
sake of reducing computation complexity).
% %
% \marginJC{(54) does not include (27a) explicitly, and based on our
%   current discussion in Section VI, we would need to add it. But if
%   you check my margins, I don't think we do if we set
%   $k\in\{0,1,\dots,n\}$.}
%   \marginXW{In the new formulation (54c) aligns with (27a). }
%
Building on this, we add (\ref{eq_simuSDP}a-b) to formulate an SDP
problem. The motivation for these two terms is as follows.
%
% \marginJC{We say ``motivation for these two terms is as follows.'' but
% we do not mention (54a) in what follows.}
% \marginXW{Added.}
%
From the inequality~\eqref{eq_lmi0}, once $P$ is given, the decrease
of the Lyapunov function depends on the eigenvalues of the negative
matrix in \eqref{eq_lmi0}. Furthermore, since
$0<\bm{r}<\frac{s}{1-\alpha}$, based on the definition of
$\overline{R}(\bm{\epsilon},\bm{r})$ under \eqref{eq_dstmodels0}, when
$\bm{\epsilon}$ is sufficiently small, we have
$\overline{R}(\bm{\epsilon},\bm{r})=I_n$. Substituting this into
\eqref{eq_lmi0} and applying the Schur complement, the decrease of the
Lyapunov function is reflected by the magnitude of positive $\gamma$
in (\ref{eq_simuSDP}c). This motivities the maximization
of $\gamma$ in (\ref{eq_simuSDP}a).
% %
% \marginJC{you mean maximization, not minimization, no?}
% \marginXW{Addressed.}
%
%After applying the Schur complement, this is reflected by the magnitude of positive $\gamma$ in (\ref{eq_simuSDP}c). 
Furthermore, since in (\ref{eq_simuSDP}c),
$\gamma$ can grow linearly (to infinity) with $\overline{P}$, we need
to normalize the magnitude of $\overline{P}$, leading to the condition
(\ref{eq_simuSDP}b).  We validate that problem~\eqref{eq_simuSDP} is
solvable, and compute the controller gains $K_1$, $K_2$. Given a
random initial state $\bm{x}(0)$,
we show in Figure~\ref{Fig_ddt} the system trajectory of
the closed-loop system.  One can observe that all the system states
are stabilized at the reference input.

\begin{figure}[htb]
  \centering
  \includegraphics[width=\linewidth]{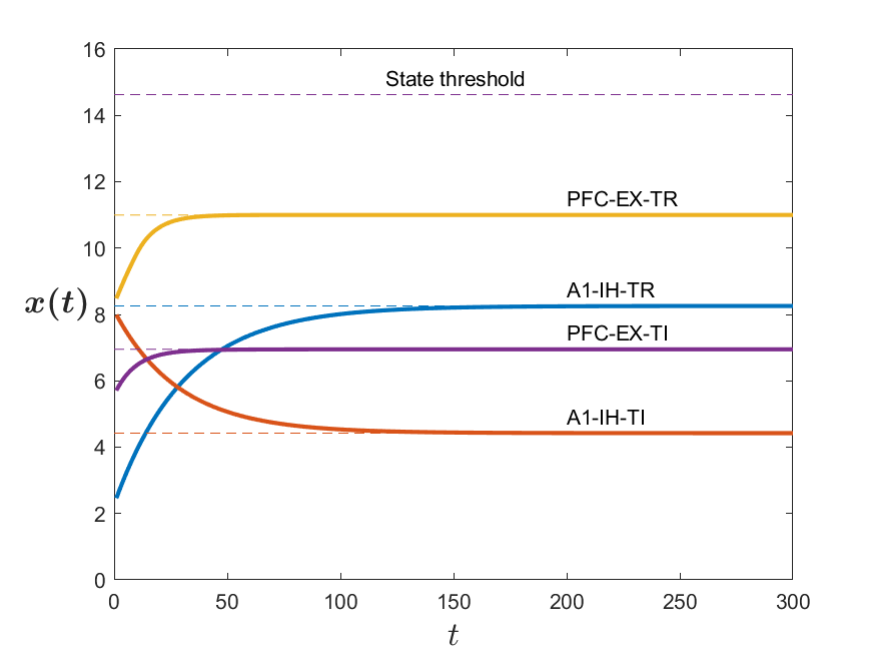}
  \caption{Data-driven stabilization of a 4-node network. The
    synthesis of the feedback gain matrix is based on solving the SDP
    specified in~\eqref{eq_simuSDP}.}
  \label{Fig_ddt}
\end{figure}
%
% \marginJC{Why not plot the references with dashed lines?}
% \marginXW{Added.}
%

To also validate Theorem~\ref{Thm_stable_INT}, we consider the
following SDP problem
\begin{subequations}\label{eq_simuSDP_INT}
  \begin{alignat}{2}
    & \text{maximize} \quad  && \gamma \label{eq_simuobj_INT}
    \\
    & \text{subject to} \quad && \overline{P}\preceq I_n
    \\
    &&& 
      \begin{bmatrix}
        \overline P
        & *
        \\
        \begin{pmatrix}
          \alpha \overline I_n \! +\!  ZS_1
          &  \alpha \overline P_{12}
            \! +\!  ZS_2
          \\ 
          \overline P_{11} \!+\! \overline P_{12}^{\top}
          &
            \overline P_{12} \!+\! \overline P_{22}
        \end{pmatrix}
        &
          \overline P
      \end{bmatrix}
        \succeq \gamma I_{4n}
    \\
    &&& 
    \begin{bmatrix}
      \overline P & *\\
      \!\begin{pmatrix}
          \alpha \overline P_{11} \!+\! \widetilde{R}_kZS_1
          &  \!\!\alpha \overline P_{12} \!+\! \widetilde{R}_kZS_2
          \\
          \overline P_{11} \!+\! \overline P_{12}^{\top}
          & \overline P_{12} \!+\! \overline P_{22}
        \end{pmatrix}
      &\!\!\! \overline P
    \end{bmatrix}\succeq 0
    \\
    &&&
    \mathcal{L}Q^{\top}S_1=\bm{0},~~~~~\mathcal{L}Q^{\top}S_2=\bm{0}
    \\
    &&&
    C_1Q^{\top}S_1=\overline P_{11},~~~~~C_1Q^{\top}S_2=\overline
    P_{12} 
  \end{alignat}
  % \begin{align}
  %   \text{maximize}\quad \gamma\quad&\label{eq_simuobj_INT}\\	
  %   \overline{P}&\preceq I_n\\  
  %   \begin{bmatrix}
  %     \overline P & *\\
  %     \begin{pmatrix}
  %       \alpha \overline I_n \! +\!  ZS_1 &  \alpha \overline P_{12}
  %                                           \! +\!  ZS_2
  %       \\ 
  %       \overline P_{11} \!+\! \overline P_{12}^{\top} & \overline P_{12} \!+\! \overline P_{22}
  %     \end{pmatrix} & \overline P
  %   \end{bmatrix}&\succeq \gamma I_{2n} \\  
  %   \begin{bmatrix}
  %     \overline P & *\\
  %     \!\begin{pmatrix}
  %         \alpha \overline P_{11} \!+\! R_iZS_1 &  \!\!\alpha \overline P_{12} \!+\! R_iZS_2\\
  %         \overline P_{11} \!+\! \overline P_{12}^{\top} & \overline P_{12} \!+\! \overline P_{22}
  %       \end{pmatrix} &\!\!\! \overline P
  %   \end{bmatrix}&\succeq 0\\
  %   \mathcal{L}Q^{\top}S_1=\bm{0},~~~~~\mathcal{L}Q^{\top}S_2=\bm{0}~~~~~~& 
  %   \\
  %   C_1Q^{\top}S_1=\overline P_{11},~~~~~C_1Q^{\top}S_2=\overline P_{12}~~~& 
  % \end{align}
\end{subequations}
for all $\widetilde{R}_k$, $k\in\{1,...,n\}$.
%
% \marginJC{In SDP, you write $R_i$ instead of $\widetilde{R}_k$. And if
% you write $k\in\{0,1,\dots,n\}$, then I think you can eliminate (55c).}
% \marginXW{Revised.}
%

This allows us to design a controller with integral feedback by
defining $\begin{bmatrix} K_1&K_2
 \end{bmatrix}
 =C_2 Q^{\top}
 \begin{bmatrix}
   S_1&S_2
 \end{bmatrix}\overline P^{-1}$
 and setting
\begin{align*}
  &\bm{u}(t)= K_1(\bm{x}(t)-\bm{r})+ K_2\bm{\xi}(t)
  \\
  &\bm{\xi}(t+1)= \bm{\xi}(t)+ (\bm{x}(t)-\bm{r}) .
\end{align*}
The motivation for this SDP is the same as \eqref{eq_simuSDP}.  We
validate that this problem is also solvable, and compute the
controller gains $K_1$, $K_2$. Given a random initial state
$\bm{x}(0)$ and a random reference value $r=
\begin{bmatrix} 8.26& 4.42& 10.99& 6.95
\end{bmatrix}^{\top}$, we show in Figure~\ref{Fig_ddt_int} the 
trajectory of 
the closed-loop system with integral feedback control.  One can
observe that all the  states 
are stabilized at the reference value. 

\begin{figure}[htb]
  \centering%
  \includegraphics[width=\linewidth]{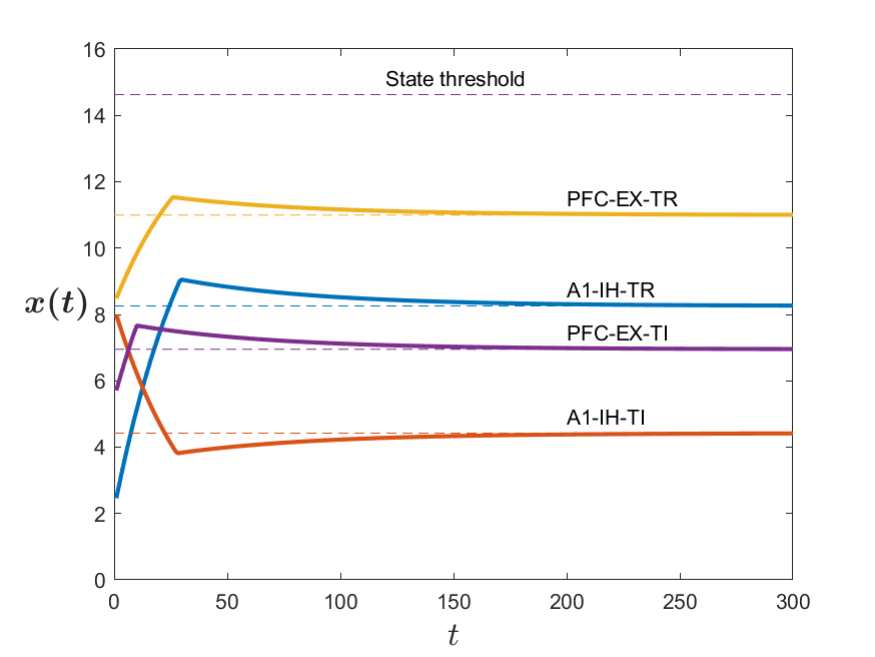}
  \caption{Data-driven stabilization of a 4-node network. The
    synthesis of the feedback gain matrix is based on solving the SDP
    specified in~\eqref{eq_simuSDP_INT}.}
  \label{Fig_ddt_int}
\end{figure}
%
% \marginJC{Plot references?}
% \marginXW{Added.}
%

Compared with Figure~\ref{Fig_ddt}, this result shows a faster response
time in general. Furthermore, we observe larger overshoots and
non-smoothness on these trajectories. This is caused by the joint
effect of the integral term and the linear-threshold function, i.e.,
when the linear-thresholds are activated, the changing rates of the
trajectories are bounded and the integral quickly accumulates; when
the linear-thresholds changes from active to deactive, a sharp,
non-smooth transition happens to the trajectory.

\subsubsection*{Effect of system disturbances}
We verify the robustness of the two controllers \eqref{eq_theproblem}
and \eqref{eq_intctrl} under random disturbances. Consider the
application of disturbance $\omega(t)\in\mathbb{R}^{4\times 1}$ to
system~\eqref{eq_dstmodel} to yield
\begin{align}\label{eq_dstmodel_noise}
  \bm{x}(t+1)=\alpha \bm{x}(t) +
  \left[W\bm{x}(t)+B\bm{u}(t)+\omega(t)\right]_0^{s}.
\end{align}
We assume $\omega(t)$ is i.i.d. and at each time step, its entries are
randomly chosen from $[0~0.2]$ following a uniform distribution. Using
the data corrupted by noise, we design controllers based on
SDPs~\eqref{eq_simuSDP} and~\eqref{eq_simuSDP_INT} to obtain
controllers that can stabilize the system states to the desired
value~$r$.
% %
% \marginJC{So these are different controllers that the ones we just
%   defined? Is this b/c we're using now data corrupted by noise?}
% \marginXW{Yes, in this case, we assume only corrupted data is available.}
% %
Figure~\ref{Fig_ddt_noise} shows that the
controller~\eqref{eq_theproblem} is not able to stabilize the system
states to the desired values. This is because the controller uses a
feed-forward mechanism to handle the reference value, which is not
able to compensate the system disturbances.

\begin{figure}[htb]
  \centering
  \includegraphics[width=\linewidth]{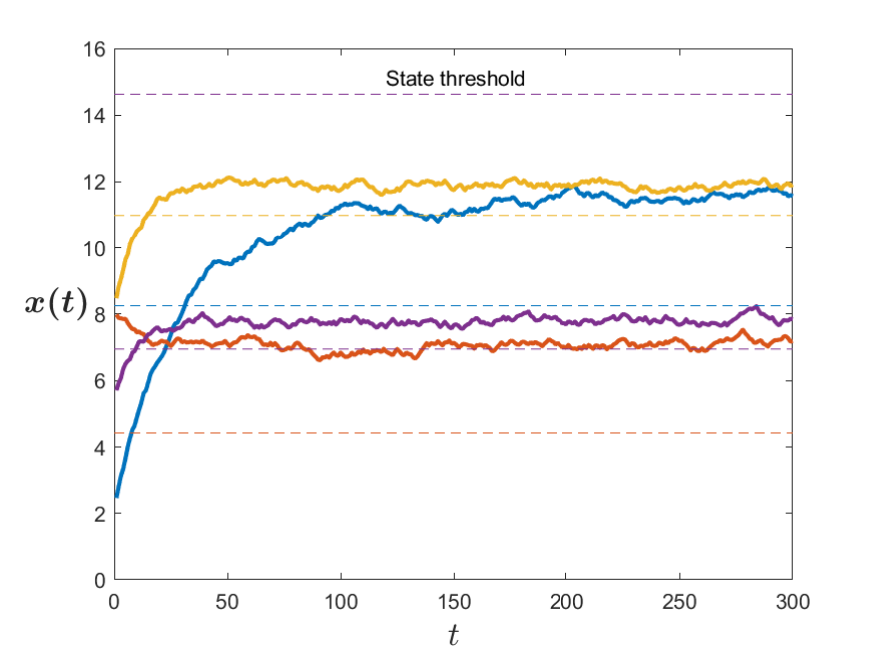}
  \caption{Data-driven stabilization of the system with controller
    \eqref{eq_theproblem} in the presence of system disturbances.
    Dashed lines correspond to the values of $r$. }
  \label{Fig_ddt_noise}
\end{figure}

In contrast, Figure~\ref{Fig_ddt_int_noise} shows that the augmented
feedback controller~\eqref{eq_simuSDP_INT} with error integration is
able to stabilize the system states to the desired values and rejects
the tracking error caused by the system disturbances.

\begin{figure}[htb]
  \centering
  \includegraphics[width=\linewidth]{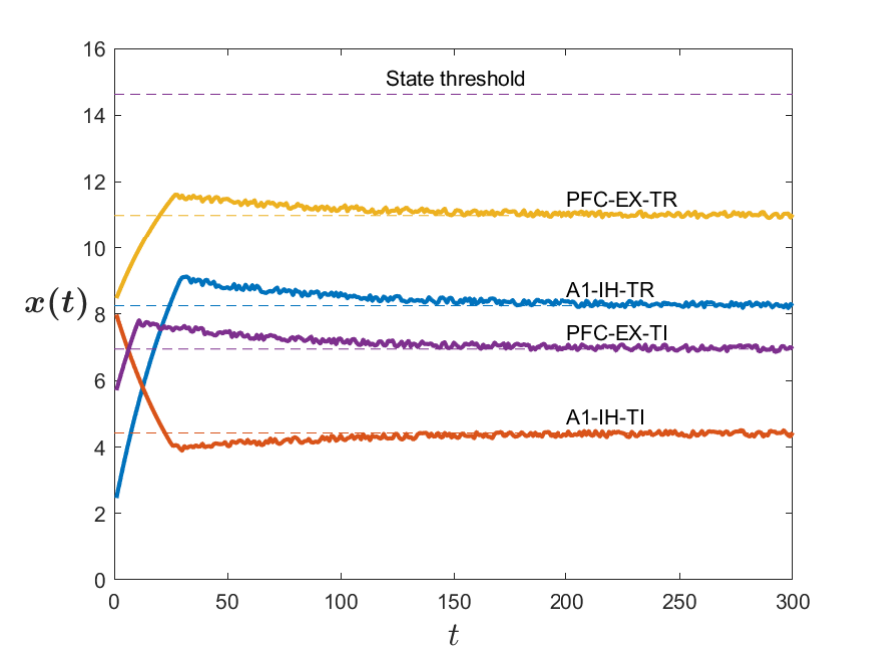}
  \caption{Data-driven stabilization of the system with controller
    \eqref{eq_simuSDP_INT} in the presence of system disturbance.
    Dashed lines correspond to the values of $r$. }
  \label{Fig_ddt_int_noise}
\end{figure}

\subsection{Data-driven Regulation of Arousal in Sensory-Motor Task}
% Paper link https://www.pnas.org/doi/epdf/10.1073/pnas.1817207116
%% EEG justified in
%% \cite{MM-EN-JC:24,EN-MAB-JS-LC-EJC-XH-ASM-GJP-DSB:20,GA-KAD-EN:24}

We validate the proposed data-driven control approach by stimulating a
similar arousal regulation experiment studied
in~\cite{JF-JC-SS-PS:19}. Arousal refers to the state of being
physiologically alert, awake, and attentive, which significantly
affects a human's ability to make decisions and take actions in
dynamic environments. According to the Yerkes–Dodson
Law~\cite{DMD-AMC-CRP-JH-PRZ:07}, there is an optimal mid-range level
of arousal for peak performance. Deviations from this optimal level
can hinder effectiveness, as shown in Figure~\ref{Fig_Arousal_vis}
(right).

\begin{figure}[htb]
  \centering
  \includegraphics[width=\linewidth]{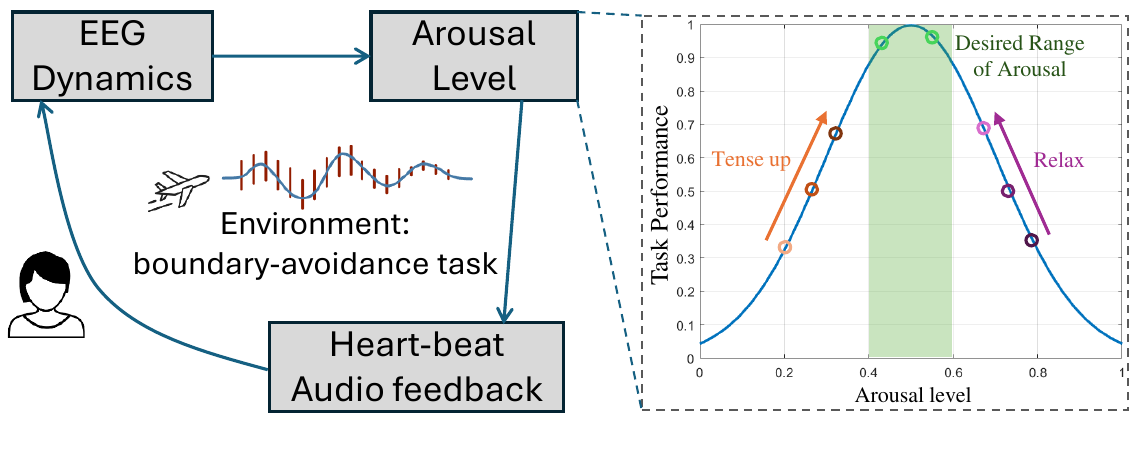}
  \caption{Left: Regulation of arousal using audio feedback in
    sensory-motor tasks. Right: The relation between arousal level and
    task performance.}
  \label{Fig_Arousal_vis}
\end{figure}

The experiment performed in \cite{JF-JC-SS-PS:19} considers a
boundary-avoidance task (BAT) paradigm, which demands the
sensory-motor responsiveness of humans. The experiment uses a virtual
reality (VR) environment, where participants navigate a plane through
courses of rectangular red waypoints (`rings') and the game fails if
the plane misses a ring.  A neurofeedback mechanism, cf.
Figure~\ref{Fig_Arousal_vis}, is employed to regulate humans' arousal
levels. Using EEG signal features (specifically, spectral
information), the system decodes a participant’s arousal level and
generates auditory feedback in the form of a heartbeat sound. The
volume of the sound increases linearly with the decoded arousal level,
which `warns' participants when their arousal level becomes
excessively high. This feedback helps participants `relax' and
maintain their arousal within an optimal range to improve task
performance.  The experiment demonstrated the effectiveness of this
feedback mechanism in regulating arousal levels.  It is worth
mentioning that in \cite{JF-JC-SS-PS:19}, both the arousal decoder and
the auditory feedback coefficient derived from arousal levels are
qualitative, since true arousal is difficult to quantify and
accurately modeling human behavior for optimal control parameters is
challenging. Consequently, the parameter choices employed there are
based on experience and experimental results. In contrast, the method
proposed in this paper is data-driven, which can circumvent these
limitations by enabling a systematic and quantitative design of the
controller.

To simulate human EEG dynamics, we construct a linear threshold model
in which the system state represents the spectral properties of the
human's EEG signal.  The spectral properties of EEG signals can be
approximated using linear-threshold networks, as justified by several
works~\cite{EN-MAB-JS-LC-EJC-XH-ASM-GJP-DSB:24,GA-KAD-EN:24}.  We
choose the dimension of the state as $x\in\mathbb{R}^{15}$, which can
be justified as follows: we consider 5 frequency-bands (delta, theta,
alpha, beta, and gamma: $[0.5, 4], [4, 8], [8, 15], [15, 24]$, and
$[24, 50]$ Hz, respectively) of the EEG signal and, for each band, we
select three dominant frequencies and normalize them, resulting in a
15-dimensional representation. This approach is similar to the
surrogate subspace introduced in~\cite{JF-JC-SS-PS:19}.  The system
input $\bm{u}\in\mathbb{R}$ represents the volume of the auditory
feedback.  To map the EEG states to the arousal level
$a_{\text{rou}}$, we introduce a linear
mapping $$a_{\text{rou}}=\phi^{\top}x\in[0\%,~100\%].$$ The model
parameters are chosen as follows: the decay rate $\alpha=0.7$ drives
the EEG activities to resting states; we set $s=0.3$, so that the
normalized states $x$ are bounded within $\frac{s}{1-\alpha}=1$.  The
state update matrix $W$ characterizes the interdependence of EEG
channels and how they change due to factors such as the difficulty of
the BAT game; the input matrix $B$ represents how humans' EEG states
are impacted by auditory feedback.  Since $W$ and $B$ can vary
significantly among different subjects and are difficult to model
explicitly, in this simulation, they are randomly generated with each
entry chosen uniformly from $[-0.5,0.5]$. However, to make sure these
matrix are meaningful for the experiment, we ensure that an increase
in $u$ (auditory feedback volume) leads to a decrease in
$a_{\text{rou}}$.  This aligns with the experiment setup
in~\cite{JF-JC-SS-PS:19}, where subjects are instructed to relax in
response to louder audio feedback.  During the experiment, the mapping
$\phi^{\top}$ that translates EEG states into arousal levels is
unknown, so the arousal level cannot be directly decoded, which is
different from~\cite{JF-JC-SS-PS:19}.  Instead, we demonstrate the
data-driven method proposed here, that uses EEG state feedback to
regulate the arousal level.

Our control objective is to drive the subject's EEG states to a target
frequency pattern $\bm{r}_T$.
% %
% \marginJC{But EEG is typically oscillatory, no? Is this $\bm{r}_T$
%   constant?}  \marginXW{The state considered here is the dominant
%   frequency of the EEG.  A constant $\bm{r}_T$ describes a target
%   frequency pattern.}
% %
Such $\bm{r}_T$ is subject-specific, associated with the desired
arousal level of the human, and can be determined during the testing
trials.  In addition, since the EEG states are subject to additive
noise, we use the controller with error integration.  We formulate the
SDP problem~\eqref{eq_simuSDP_INT} and solve for the controller gains
$K_1$ and $K_2$, which maps the EEG spectral properties $\bm{x}$ to a
scaler volume of the auditory feedback
input~$\bm{u}$.  %The convergence of system states is difficult to display
%due to the high dimensionality of the 64 EEG channels, so we compute a
%arousal level-decoder, $f(\bm{x}) = a_{\text{rou}}\in\mathbb{R}$, as
%in~\cite{JF-JC-SS-PS:19}. This decoder uses shrinkage regularized
%linear discriminant analysis to map EEG states to arousal level index
%between 0 and 100\%. 
Figure~\ref{Fig_Arousal_traj} (left) demonstrates
the effectiveness of the obtained data-driven controller in regulating
the arousal into a desired range from an overly tensed-up state and
maintaining it in this region. The overshoot of the trajectory is caused
by the integral term.

\begin{figure}[htb]
  \centering \includegraphics[width=\linewidth]{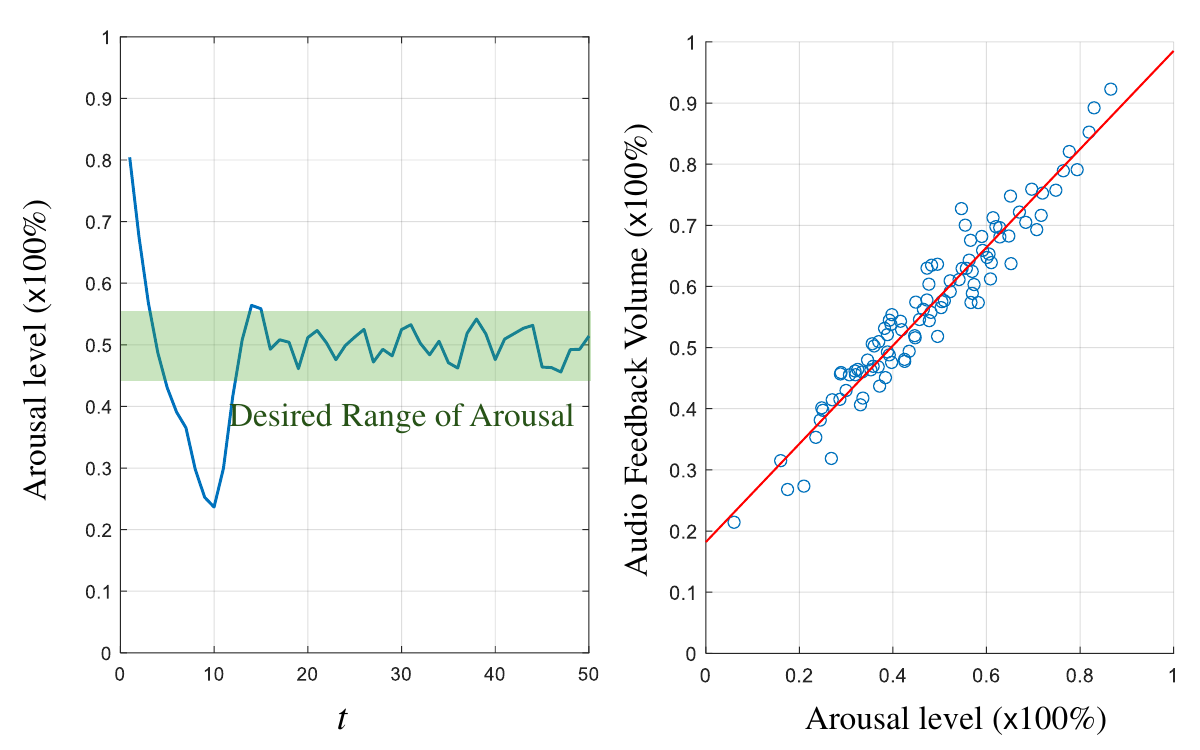}
  \caption{Left: regulation of arousal using the proposed data-driven
    controller. Right: relationship between arousal level and audio
    feedback volume based on the state feedback matrix~$K_1$.}
  \label{Fig_Arousal_traj}
\end{figure}

Furthermore, Figure~\ref{Fig_Arousal_traj} (right), demonstrates the
relationship between the arousal level and the audio feedback
volume. This is done by taking random EEG states $\bm{x}$ and
observing the relations between $a_{\text{rou}}=\phi^{\top}x$ and
$u=K_1\bm{x}$. Here, since the integrator $\xi$ 
is dynamically changing, except at equilibrium, its impact on the control 
input $u$ is difficult to characterize. Therefore, we set $\xi=0$ and 
ignore the $K_2\xi$ term in $u$. 
% %
% \marginJC{We seem to repeat random $\bm{x}$ in the same
%   sentence. Also, we said that we were using (55) but only write here
%   $K_1\bm{x}$. For which $\xi$? Fixed for all random $\bm{x}$?}
%   \marginXW{It has been ignored and set $\xi=0$.}
% %
The linear regression of the relation between $a_{\text{rou}}$ and $u$
(considering only the proportional term $K_1$)
suggests that the volume should increase with the arousal level of the
human, i.e., a higher volume reminds the human that their arousal
level is high, and the human will try to calm down to reduce their
arousal level. Our result hence aligns with the experiment design
in~\cite{JF-JC-SS-PS:19} and provides a justification for this
feedback mechanism.

\section{Conclusions and Future Work}
We have designed data-driven controllers to stabilize unknown
linear-threshold network models to a given reference value.
Exploiting the special structure of the linear-threshold model, we
have established a data-based representation of the dynamics relying
on a map that reconstructs the system's state-input datasets.
Building on this, we have obtained closed-loop data-based
representations for two types of data-driven controllers: state
feedback with feed-forward reference input and augmented feedback
controller with error integration.  In both cases, we have shown how
to combine these representations with techniques from switched systems
theory to identify stabilization conditions in the form of a set of
linear matrix inequalities (LMIs), whose solutions correspond to the
controller gain matrices. We have formally established the correctness
of the proposed designs. Given that the complexity of the LMI
formulations grows exponentially with the system state, we have
proposed alternative sufficient conditions to solve the LMIs that
scale linearly without sacrificing performance.  We have validated the
effectiveness of the two controllers in two different case studies.
Future work will investigate controller designs beyond time-invariant
ones, extend our results to scenarios where access to full-state
information is not available, and explore the applicability of the
results to other case studies.

% apply the proposed approach to real neural/social network systems.

%\section*{Acknowledgments}
%This work was partially supported by NSF ECCS Award 2332210, NSF CMMI Award
%2308640, and MURI ARO Award W911NF-24-1-0228.
  
% \marginJC{Add acknowledgement to your colleague for suggesting the
% arousal experiment? Also, we could acknowledge funding agencies. On
% my side, ``This work was partially supported by NSF CMMI Award
% 2308640 and MURI ARO Award W911NF-24-1-0228''. }

% \marginJC{Have any of the arxiv entries already appeared? If so,
%   update in bib.}

\bibliographystyle{IEEEtran}
\bibliography{bib/alias,bib/New,bib/Main-add,bib/JC,bib/Main}  
% \bibliography{alias, New, Main-add, JC, Main}

\begin{IEEEbiography}[{\includegraphics[width=1.0
    in,height=1.35in,clip,keepaspectratio]{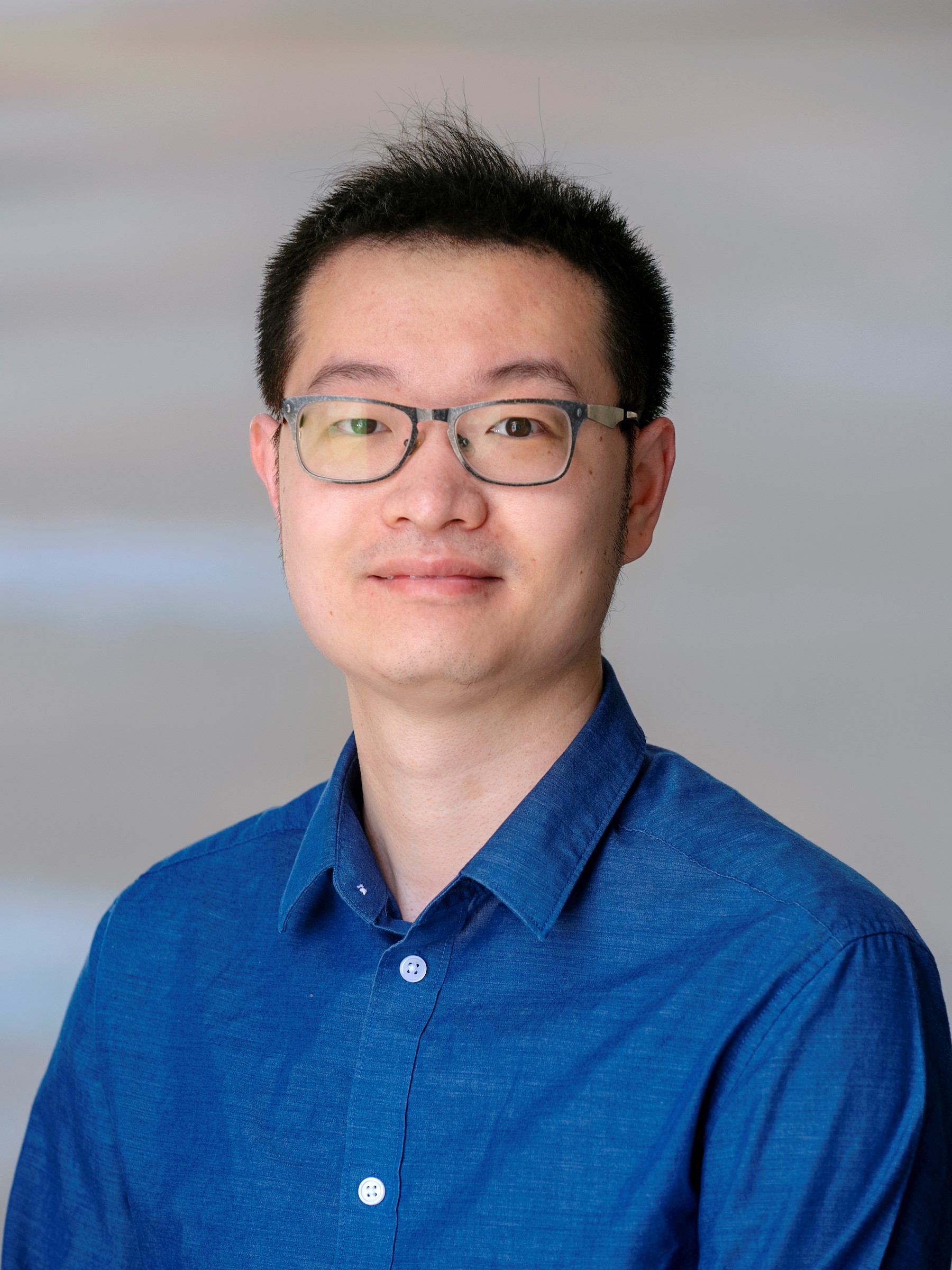}}]{Xuan Wang}
  is an Assistant Professor with the Department of Electrical and
  Computer Engineering at George Mason University.  He received his
  Ph.D. degree in autonomy and control, from the School of Aeronautics
  and Astronautics, Purdue University in 2020. He was a post-doctoral
  researcher with the Department of Mechanical and Aerospace
  Engineering at the University of California, San Diego from 2020 to
  2021. His research interests include multi-agent control and
  optimization; resilient multi-agent coordination; system
  identification and data-driven control of network dynamical systems.
\end{IEEEbiography}

\begin{IEEEbiography}[{\includegraphics[width=1.0
    in,height=1.35in,clip,keepaspectratio]{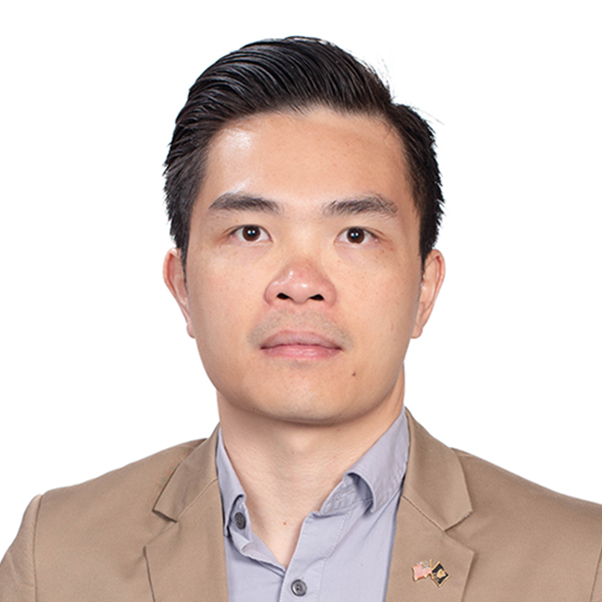}}]{Duy Duong-Tran}
is an affiliated faculty at the Perelman School of Medicine at the University of Pennsylvania. Before that, he was a Post-doctoral Research Fellow at the school. He held a Ph.D. from Purdue University’s School of Industrial Engineering and a graduate certificate
from Purdue’s School of Engineering Education in 2022. His main research interest is at the crossroads between data science, computational neuroscience, engineering education, and biomedical NLP.
\end{IEEEbiography}

\begin{IEEEbiography}[{\includegraphics[width=1in,height=1.25in,clip,keepaspectratio]{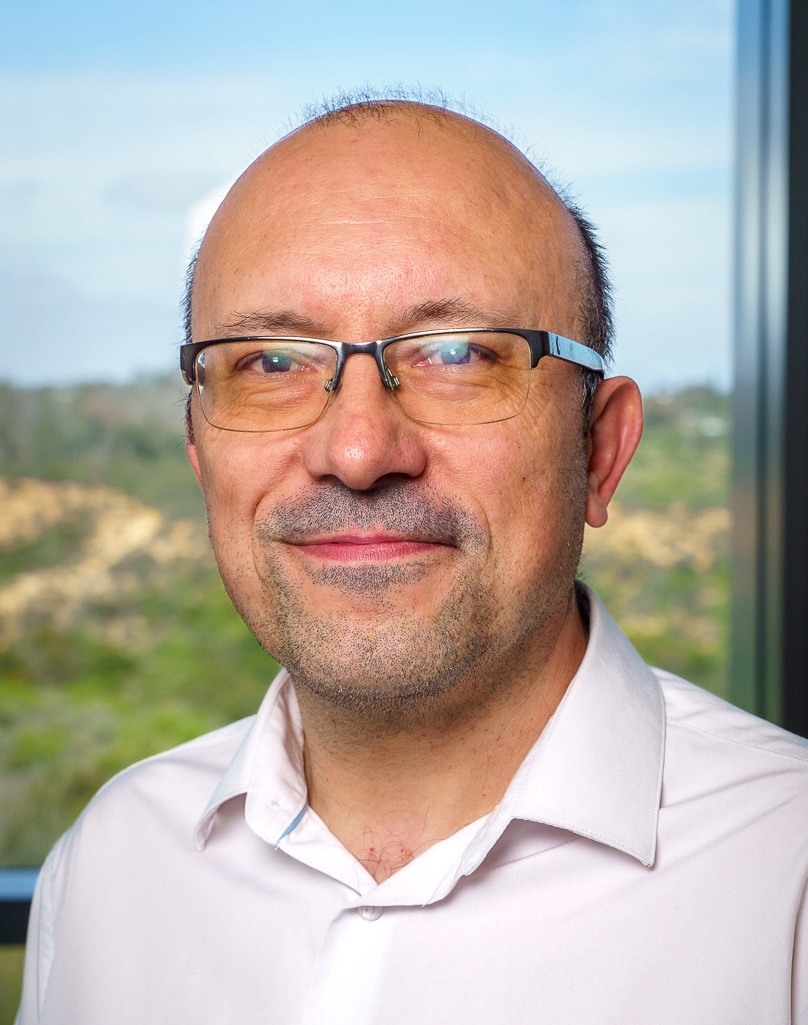}}]{Jorge
    Cort\'{e}s}(M'02, SM'06, F'14) received the Licenciatura degree in
  mathematics from Universidad de Zaragoza, Zaragoza, Spain, in 1997,
  and the Ph.D. degree in engineering mathematics from Universidad
  Carlos III de Madrid, Madrid, Spain, in 2001. He held postdoctoral
  positions with the University of Twente, Twente, The Netherlands,
  and the University of Illinois at Urbana-Champaign, Urbana, IL,
  USA. He was an Assistant Professor with the Department of Applied
  Mathematics and Statistics, University of California, Santa Cruz,
  CA, USA, from 2004 to 2007. He is a Professor and Cymer Corporation
  Endowed Chair in High Performance Dynamic Systems Modeling and
  Control at the Department of Mechanical and Aerospace Engineering,
  University of California, San Diego, CA, USA.  He is a Fellow of
  IEEE, SIAM, and IFAC.  His research interests include distributed
  control and optimization, network science, nonsmooth analysis,
  reasoning and decision making under uncertainty, network
  neuroscience, and multi-agent coordination in robotic, power, and
  transportation networks.
\end{IEEEbiography}

\end{document}